\documentclass[12pt,a4paper,american]{article}

\usepackage{amsmath}
\usepackage{graphicx}
\usepackage{amssymb}

\usepackage{amsthm}
\usepackage{amsfonts}
\usepackage{mathrsfs}

\newcommand{\lyxaddress}[1]{
\par {\raggedright #1
\vspace{1.4em}
\noindent\par}
}

\numberwithin{equation}{section}

\newtheorem{theorem}{Theorem}[section]
\newtheorem{proposition}[theorem]{Proposition}
\newtheorem{lemma}[theorem]{Lemma}

\newtheorem{conjecture}[theorem]{Conjecture}
\newtheorem{remark}[theorem]{Remark}
\theoremstyle{remark}

\newcommand{\sF}{\mathscr{F}}

\renewcommand{\Im}{\mathop\mathrm{Im}\nolimits}
\renewcommand{\Re}{\mathop\mathrm{Re}\nolimits}
\newcommand{\supp}{\mathop\mathrm{supp}\nolimits}
\newcommand{\dist}{\mathop\mathrm{dist}\nolimits}
\newcommand{\dd}{\mathrm{d}}
\newcommand{\const}{\mathrm{const}}

\begin{document}

\title{Resonant cyclotron acceleration of particles\\
  by a time periodic singular flux tube}

\date{{}}

\author{J.~Asch$^{1}$, T.~Kalvoda$^{2}$, P.~\v{S}\v{t}ov\'\i\v{c}ek$^{3}$}

\maketitle

\lyxaddress{$^{1}$Centre de Physique Th\'eorique (CPT-CNRS UMR
    6207) Universit\'e du Sud, Toulon-Var, BP~20132, F--83957 La Garde
    Cedex, France ({\tt asch@cpt.univ-mrs.fr}).\newline
    Supported by Grant PHC-Barrande No. 21899QB, Minist\`ere des
    affaires \'etrang\`eres et europ\'eennes.}
\lyxaddress{$^{2}$Department of Theoretical Computer Science,
    Faculty of Information Technology, Czech Technical University
    in~Prague, Kolejni~2, 120~00 Praha, Czech Republic
    ({\tt tomas.kalvoda@fit.cvut.cz}).\newline
    Supported by Grant No. LC06002 of the Ministry of Education of
    the Czech Republic and by Grant No. 202/08/H072 of the Czech
    Science Foundation.}
\lyxaddress{$^{3}$Department of Mathematics,
    Faculty of Nuclear Science, Czech Technical University in~Prague,
    Trojanova 13, 120~00 Praha, Czech Republic
    ({\tt stovicek@kmlinux.fjfi.cvut.cz}).\newline
    Supported by Grant No.\ 201/09/0811 of the Czech Science
    Foundation.}

\begin{abstract}
  \noindent We study the dynamics of a classical nonrelativistic
  charged particle moving on a punctured plane under the influence of
  a homogeneous magnetic field and driven by a periodically
  time-dependent singular flux tube through the hole. We observe an
  effect of resonance of the flux and cyclotron frequencies. The
  particle is accelerated to arbitrarily high energies even by a flux
  of small field strength which is not necessarily encircled by the
  cyclotron orbit; the cyclotron orbits blow up and the particle
  oscillates between the hole and infinity.  We support this
  observation by an analytic study of an approximation for small
  amplitudes of the flux which is obtained with the aid of averaging
  methods. This way we derive asymptotic formulas that are afterwards
  shown to represent a good description of the accelerated motion even
  for fluxes which are not necessarily small. More precisely, we argue
  that the leading asymptotic terms may be regarded as approximate
  solutions of the original system in the asymptotic domain as the
  time tends to infinity.\\


\end{abstract}

\section{Introduction}

Consider a classical point particle of mass $m$ and charge $e$ moving
on the punctured plane $\mathbb{R}^{2}\setminus\{0\}$ in the presence
of a homogeneous magnetic field of magnitude $b$. Suppose further that
a singular magnetic flux line whose strength $\Phi(t)$ is oscillating
with frequency $\Omega$ intersects the plane at the origin. The
equations of motion in phase space
$\mathbb{P}=(\mathbb{R}^{2}\setminus\{0\})\times\mathbb{R}^{2}$ are
generated by the time--dependent Hamiltonian
\begin{equation}
  \label{eq:H_cartesian_def}
  H(q,p,t)=\frac{1}{2m}\left(p-eA(q,t)\right)^{2},\ \mbox{with\ }
  A(q,t)=\left(-\frac{b}{2}
    +\frac{\Phi(t)}{2\pi\vert q\vert^{2}}\right)\! q^{\perp},
\end{equation}
where $(q,p)\in\mathbb{P}$, $t\in\mathbb{R}$. Here and throughout we
denote $x^{\perp}=(-x_{2},x_{1})$ for
$x=(x_{1},x_{2})\in\mathbb{R}^{2}$. Our aim is to understand the
dynamics of this system for large times. Of particular interest is the
growth of energy as well as the drift of the guiding center.

Our main result in the present paper is to exhibit and to prove a
resonance effect whose origin can qualitatively be understood as
follows.  The Lorentz force equals
\[
eb\, q'(t)^{\perp}+eE(t)\quad\mbox{where}\quad E(t)
=-\frac{\Phi'(t)}{2\pi|q|^{2}}\, q^{\perp}.
\]
In the induction free case when $\Phi'(t)=0$, the particle moves along
a circle of fixed center (cyclotron orbit) with the cyclotron
frequency $\omega_{c}=\left\vert eb\right\vert /m$ and the cyclotron
radius depending bijectively on the energy. If $\Phi'(t)$ is nonzero
but small then the energy, the frequency and the guiding center of the
orbit become time dependent. For the time derivative of the energy
computed in polar coordinates $(r,\theta)$ one finds that
\begin{equation}
  \frac{\mathrm{d}}{\mathrm{d}t}H
  = \frac{\partial}{\partial t}H
  = e\, q'(t)\cdot E(t)=-\frac{e}{2\pi}\,\theta'(t)\Phi'(t)
  \label{eq:D_der}
\end{equation}
and so the acceleration rate is given by
\begin{equation}
  \gamma = \lim_{\tau\to\infty}\frac{H(\tau)}{\tau}
  = -\frac{e}{2\pi}\,\lim_{\tau\to\infty}\,
  \frac{1}{\tau}\int_{0}^{\tau}\theta'(t)\Phi'(t)dt.
  \label{eq:rate}
\end{equation}
One may expect that $\theta'(t)$ is close to a periodic function of
frequency $\omega_{c}$ (which is the case when $\Phi'(t)=0$).  If the
frequencies $\omega_{c}$ and $\Omega$ are resonant, i.e. if there
exist indices $n$ and $m$ in the support of the Fourier transforms of
$\theta'(t)$ and $\Phi'(t)$, respectively, such that
$n\omega_{c}+m\Omega=0$ then one may speculate that $\gamma$ is
positive.

Our original motivation to study this problem was to understand the
dynamics of so called quantum charge pumps \cite{Laughlin,Niu} already
on a classical level and to gain a detailed intuition on the dynamical
mechanisms in the simplest case.  We suggest, however, that particle
acceleration mechanisms might be of actual interest in various other
models, for example in interstellar physics \cite{Bolton,Horne}.  We
speculate that the equations of motion in this domain might exhibit
similar ingredients as the ones studied here.

We remark, too, that our results could be of interest in accelerator
physics. While the betatron principle uses a linearly time dependent
flux tube to accelerate particles on cyclotron orbits around the flux
\cite{KerstSerber}, the resonance effect we observe in the present
paper has the feature that acceleration can be achieved with
arbitrarily small field strength. A second aspect is that, in contrast
to the case of a linearly increasing flux, cyclotron orbits which do
not encircle the flux tube are accelerated as well. In fact, for the
linear case it was shown in \cite{AschStovicek} that outside the flux
tube one has the usual drift of the guiding center, without
acceleration, along the field lines of the averaged potential.

In the case of resonant frequencies one readily observes an
accelerated motion numerically, see Figure~\ref{spiral} for an
illustration. A typical resonant trajectory is a helix-like curve
whose center goes out at the same rate as the radius grows. At the
moment, a complete treatment of the equations of motion is out of
reach. Therefore we first apply a resonant averaging method to derive
a Hamiltonian which is formally a first order approximation in a small
coupling constant of the flux tube. We then analyze its flow and show
a certain type of asymptotic behavior at large times with the aid of
differential topological methods. Though this approach gives only an
existence result we use the analysis to derive the leading asymptotic
terms.

Then, in the second stage of our analysis, these asymptotic terms are
used to formulate a simplified version of the evolution equations. We
decouple the equations by substituting the anticipated leading
asymptotic terms into the right-hand sides.  The decoupled system
admits a rigorous asymptotic analysis whose conclusions turn out to be
fully consistent with the formulas derived by the averaging methods.
This holds under the assumption that the singular flux tube has a
correct sign of the time derivative when the particle passes next to
the origin. It should be emphasized that the singular magnetic flux is
not assumed to be very small in the analysis of the decoupled system.

The paper is organized as follows. In Section~\ref{sec:results} we
summarize our main results. Namely, in
Subsection~\ref{results_Hamiltonian} we introduce the Hamiltonian
equations of motion in action-angle coordinates, in
Subsection~\ref{sec:results_averaged} we study the averaged dynamics
resulting from the Poincar\'e-von~Zeipel elimination method and in
Subsection~\ref{sec:results_asymptotic} we rederive the leading
asymptotic terms by introducing a decoupled system. Finally, in
Subsection~\ref{sec:results_guiding} we interpret the formulas derived
so far in guiding center coordinates. Proofs, derivations and
additional details are postponed to Sections~\ref{sec:TheModel},
\ref{sec:vonZeipelMethod}, \ref{sec:SimplifiedEquation} and
\ref{sec:Conclusion}, respectively.

\begin{figure}

\centering\includegraphics[scale=0.55]{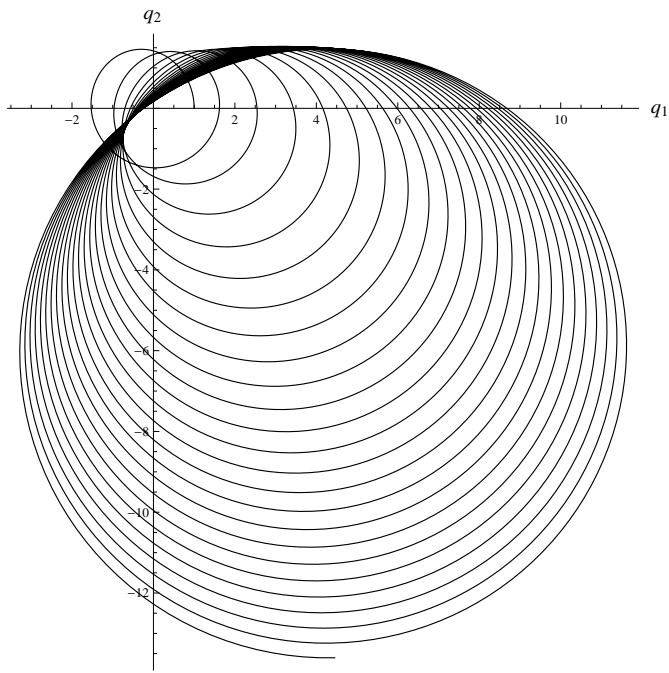}
\label{spiral}
\caption{The solution $q(t)$ of the equations of motion in the plane
  for $t\in[\,0,150\,]$, with $\Phi(t)=2\pi\epsilon f(\Omega t)$,
  $f(t)=\sin(t)-(1/3)\cos(2t)$, for the values of parameters
  $\epsilon=0.35$, $b=1$, $\Omega=1$, and with the initial conditions
  $q(0)=(1,0)$, $q'(0)=(0,1.617)$.}
\end{figure}

\section{Main results}
\label{sec:results}

\subsection{The Hamiltonian}
\label{results_Hamiltonian}

In view of the rotational symmetry of the problem we prefer to work
with polar coordinates, $q=r(\cos\theta,\sin\theta)$. Denoting by
$p_r$ and $p_\theta$ the momenta conjugate to $r$ and $\theta$,
respectively, one has
\begin{displaymath}
  p_r = p\cdot q/|q|,~
  p_\theta = p\cdot q^\perp \equiv q\wedge p.
\end{displaymath}
Using that the vector potential $A(q,t)$ is proportional to $q^\perp$
one finds for the Hamiltonian (\ref{eq:H_cartesian_def}) in polar
coordinates
\begin{equation}
  H(r,\theta,p_{r},p_{\theta},t)
  =\frac{1}{2m}\left(p_{r}^{\,2}
    +\left(\frac{1}{r}\left(p_{\theta}-\frac{e\Phi(t)}{2\pi}\right)
      +\frac{eb}{2}\, r\right)^{\!2}\right).
  \label{eq:Ham_polar_coord}
\end{equation}
The angular momentum $p_{\theta}$ is an integral of motion and thus
the analysis of the system effectively reduces to a one-dimensional
radial motion with time-dependent coefficients. From now on we set
$e=m=1$, and so the cyclotron frequency equals $b$. Put
\begin{equation}
  a(t) = p_{\theta}-\frac{1}{2\pi}\,\Phi(t).
  \label{eq:DefAt}
\end{equation}
In the radial Hamiltonian one may omit the term $ba(t)/2$ not
contributing to the equations of motion and thus one arrives at the
expression
\begin{equation}
  H_{\mathrm{rad}}(r,p_{r},t)
  = \frac{p_{r}^{\,2}}{2}+\frac{a(t)^{2}}{2r^{2}}+\frac{b^{2}r^{2}}{8}.
  \label{eq:DefHradial}
\end{equation}
For definiteness we assume that $b>0$.

The time-independent Hamiltonian system, with $a=\const$, is
explicitly solvable. In particular, one finds the corresponding
action-angle coordinates $(I,\varphi)$ depending on $a$ as a
parameter. Substituting the given function $a(t)$ for $a$ one gets a
time-dependent transformation of coordinates. This is an essential
step in the analysis of the time-dependent Hamiltonian
(\ref{eq:DefHradial}) since action-angle coordinates are appropriate
for employing averaging methods. Postponing the derivation to
Section~\ref{sec:TheModel}, here we give the transformation equations:
\begin{eqnarray}
  r & = & \frac{2}{\sqrt{b}}\left(I+\frac{|a|}{2}+\sqrt{I(I+|a|)}\,
    \sin(\varphi)\right)^{\!1/2},\label{eq:canonic_r}\\
  \noalign{\medskip}p_{r}
  & = & \frac{\sqrt{bI(I+|a|)}\,
    \cos(\varphi)}{\left(I+\frac{|a|}{2}+\sqrt{I(I+|a|)}\,
      \sin(\varphi)\right)^{\!1/2}}\,,
  \label{eq:canonic_pr}
\end{eqnarray}
and conversely,
\begin{eqnarray}
  I &=& \frac{1}{2b}\!\left(p_{r}^{\,2}+\left(\frac{|a(t)|}{r}
      -\frac{br}{2}\right)^{\!2}\right)\,,
  \label{eq:I_fce_polar}\\
  \noalign{\smallskip}
  \tan(\varphi) &=& -\frac{1}{brp_r}
  \left(p_r^{\,2}+\frac{a(t)^2}{r^2}-\frac{b^2r^2}{4}\right)\!.
  \label{eq:varphi_fce_polar}
\end{eqnarray}
Furthermore, expressing the Hamiltonian in action-angle coordinates
one obtains
\begin{equation}
  \label{eq:H_c_def}
  H_{c}(\varphi,I,t)
  = bI-|a(t)|'\arctan\!\left(\frac{\sqrt{I}
      \cos(\varphi)}{\sqrt{I+|a(t)|}+\sqrt{I}\sin(\varphi)}\right)\,,
\end{equation}
and the corresponding Hamiltonian equations of motion take the form
\begin{eqnarray}
  \varphi' & = & b-\frac{\cos(\varphi)aa'}{2\sqrt{I(I+|a|)}
    \left(2I+|a|+2\sqrt{I(I+|a|)}\,\sin(\varphi)\right)}\,,
  \label{eqmotion1}\\
  I' & = & -\frac{|a|'}{2}\left(1-\frac{|a|}{2I+|a|+2\sqrt{I(I+|a|)}\,
      \sin(\varphi)}\right).
  \label{eqmotion2}
\end{eqnarray}

Using action-angle coordinates one can give a rough qualitative
description of trajectories in the resonant case. From
(\ref{eq:canonic_r}) we see that
\[
r^{2} = \frac{1}{2}\,(r_{+}^{\;2}+r_{-}^{\;2})
+\frac{1}{2}\,(r_{+}^{\;2}-r_{-}^{\;2})\sin(\varphi)
\]
where
\begin{displaymath}
  r_{\pm}= \sqrt{\frac{2}{b}}\left(\sqrt{I+|a|}\pm\sqrt{I}\right)
\end{displaymath}
are extremal points of the trajectory (see
Section~\ref{sec:TheModel}). As formulated more precisely in the
sequel, resonant trajectories are characterized by a linear increase
of $I(t)$ for large times while $\varphi(t)\approx{}bt+\const$. Thus,
as the angle $\varphi$ increases the radius $r$ oscillates between
$r_{-}$ and $r_{+}$ (though $r_{-}$, $r_{+}$ themselves are also
time-dependent). Moreover, since $a(t)$ is bounded, $I\to\infty$
implies $r_{+}\to\infty$ and $r_{-}(t)=2|a(t)|/(b\,r_{+}(t))\to0$ as
$t\to+\infty$. Therefore the trajectory in the $q$-plane periodically
returns to the origin and then again escapes far away from it while
its extremal distances to the origin converge respectively to zero and
infinity. We refer again to Figure~\ref{spiral} for a typical
trajectory in the $q$-plane in the case of resonant frequencies.

Concluding this subsection let us make more precise some assumptions.
For the sake of definiteness we shall focus on the case when
$p_{\theta}$ is positive and greater than the amplitude of
$\Phi(t)/(2\pi)$ and so $a(t)$ is an everywhere strictly positive
function. Let us stress, however, that this restriction is not
essential for the resonance effect as the radial Hamiltonian
(\ref{eq:DefHradial}) depends only on $a(t)^{2}$ and thus the sign of
$a(t)$ is irrelevant for the motion in the radial direction. On the
other hand, as discussed in Subsection~\ref{sec:results_guiding}, the
sign of $a(t)$ determines whether the orbit encircles the singular
magnetic flux or not.

Furthermore, the frequency $\Omega>0$ of the singular flux tube is
treated as a parameter of the model and so we write
\begin{equation}
  \Phi(t)=2\pi\epsilon f(\Omega t)
  \label{eq:Phi_rel_f}
\end{equation}
where $f$ is a $2\pi$-periodic real function possibly obeying
additional assumptions. Hence
\begin{displaymath}
  a(t)=p_{\theta}-\epsilon f(\Omega t)
\end{displaymath}
where the coupling constant $\epsilon$ is supposed to be positive and,
if desired, playing the role of a small parameter.

\subsection{The dynamics generated by the first order averaged
  Hamiltonian}
\label{sec:results_averaged}

In order to study occurrences of resonant behavior we apply the
Poincar\'e-von~Zeipel elimination method which takes into account
possible resonances, as explained in detail, for instance, in
\cite{ArnoldKozlovNeishtadt}. In
Proposition~\ref{thm:ConclusionResonant} we provide a detailed
information on the resonance effect for the dynamics generated by the
first-order averaged Hamiltonian.

We start from introducing some basic notation. Let
$\mathbb{T}^{d}=(\mathbb{R}/2\pi\mathbb{Z})^{d}$ be the
$d$-dimensional torus. For $f(\varphi)\in C(\mathbb{T}^{d})$ and
$k\in\mathbb{Z}^{d}$ we denote the $k$th Fourier coefficient of $f$ by
the symbol
\[
\sF[f(\varphi)]_{k}=\frac{1}{(2\pi)^{d}}\,
\int_{\mathbb{T}^{d}}f(\varphi)\,
e^{-ik\cdot\varphi}\,\mathrm{d}\varphi.
\]
We introduce $\supp\sF[f]$ as the set of indices corresponding to
nonzero Fourier coefficients of $f(\varphi)$. For
$f\in{}C(\mathbb{T}^{d})$ and $\mathbb{L}\subset\mathbb{Z}^{d}$ put
\begin{displaymath}
  \langle f(\varphi)\rangle_{\mathbb{L}}
  = \sum_{k\in\mathbb{L}}\sF[f]_{k}\, e^{ik\cdot\varphi}.
\end{displaymath}

For example, in the formulation of
Proposition~\ref{thm:ConclusionResonant} we use the averaged function
$\langle{}f(\varphi)\rangle_{\mathbb{Z}\nu}$, with $\nu\in\mathbb{N}$.
Assuming that $f\in{}C(\mathbb{T}^{1})$ let us note that it can
alternatively be expressed, without directly referring to the Fourier
series, as
\begin{equation}
  \label{eq:f_aver_Znu}
  \langle f(\varphi)\rangle_{\mathbb{Z}\nu}
  = \frac{1}{\nu}\sum_{j=0}^{\nu-1}
  f\!\left(\varphi+\frac{2\pi}{\nu}\, j\right).
\end{equation}
To verify the formula observe that the RHS of (\ref{eq:f_aver_Znu})
again belongs to $C(\mathbb{T}^{1})$. To check that its Fourier series
coincides with the LHS of (\ref{eq:f_aver_Znu}) it suffices to
consider exponential functions $f(\varphi)=e^{ik\varphi}$, with
$k\in\mathbb{Z}$. In that case one finds that the RHS of
(\ref{eq:f_aver_Znu}) reproduces the function $f(\varphi)$ if $\nu$ is
a divisor of $k$ and vanishes otherwise.

In this subsection we assume that $\Phi(t)$ is given by
(\ref{eq:Phi_rel_f}) where $\epsilon>0$ is regarded as a small
parameter and the $2\pi$-periodic real function $f(\varphi)$ fulfills
\begin{equation}
  \sum_{k=1}^{\infty}k\,|\sF[f(\varphi)]_{k}| < \infty.
  \label{eq:SumkFfk}
\end{equation}
This implies that $f\in C^{1}(\mathbb{T}^{1})$.

Following a standard approach to time-dependent Hamiltonian systems we
pass to the extended phase space by introducing a new phase
$\varphi_{2}=\Omega t$ and its conjugate momentum $I_{2}$ thus
obtaining an equivalent autonomous system on a larger space.  To have
a unified notation we rename the old variables $\varphi$, $I$ as
$\varphi_{1}$, $I_{1}$, respectively, and set $\omega_{1}=b$,
$\omega_{2}=\Omega$.  With this new notation, we write
$I=(I_{1},I_{2})$, $\varphi=(\varphi_{1},\varphi_{2})$ (changing this
way the meaning of the symbols $I$ and $\varphi$ in the current
subsection).

The Hamiltonian on the extended phase space is defined as
\begin{displaymath}
 K(\varphi_{1},\varphi_{2},I_{1},I_{2})
 = \omega_2 I_{2}+H_{c}(\varphi_{1},I_{1},\varphi_{2}/\omega_2)
\end{displaymath}
Recalling (\ref{eq:H_c_def}) and the conventions for $p_{\theta}$ and
$\Phi$ one gets
\begin{equation}
  K(\varphi,I,\epsilon) = \omega_{1}I_{1}+\omega_{2}I_{2}
  +\epsilon F(\varphi,I,\epsilon)
  \label{eq:K_extended}
\end{equation}
where
\begin{equation}
  F(\varphi,I,\epsilon) = \omega_{2}f'(\varphi_{2})
  \arctan\!\left(\frac{\sqrt{I_{1}}
      \cos(\varphi_{1})}{\sqrt{I_{1}+p_{\theta}-\epsilon f(\varphi_{2})}
      +\sqrt{I_{1}}\sin(\varphi_{1})}\right)\!.
  \label{eq:F_vphi_I_eps}
\end{equation}

Our discussion focuses on the resonant case when
\begin{equation}
  \lambda:=\frac{\omega_{2}}{\omega_{1}}
  =\frac{\mu}{\nu}\,,\ \mbox{with\ }
  \mu,\nu\in\mathbb{N}\mbox{\ coprime},
  \label{eq:w2_to_w1_rational}
\end{equation}
and, moreover, $\nu$ fulfills
\begin{equation}
  \supp\sF[f]\cap\left(\mathbb{Z}\nu\setminus\{0\}\right)\neq\emptyset.
  \label{eq:ResonantCase}
\end{equation}
Note that (\ref{eq:ResonantCase}) happens if and only if $\langle
f(\varphi)\rangle_{\mathbb{Z}\nu}$ is not a constant function.
Discussion of the nonresonant case is, at least on the level of the
averaged dynamics, simple and we avoid it. Let us just remark that, in
that case, it is not difficult to see that trajectories as well as the
energy for the averaged Hamiltonian are bounded.

The Hamiltonian (\ref{eq:K_extended}) is appropriate for application
of the Poincar\'e-von~Zeipel elimination method whose basic scheme is
briefly recalled in Subsection~\ref{sec:vonZeipel_Notation}. The idea
is to eliminate from $K(\varphi,I,\epsilon)$, with the aid of a
canonical transformation, a subgroup of angle variables which are
classified as nonresonant, and thus to arrive at a new reduced
Hamiltonian $\mathcal{K}(\psi,J,\epsilon)$ depending only on the so
called resonant angle variables and the corresponding actions. Note,
however, that a good deal of information about the system is contained
in the canonical transformation itself.

To achieve this goal one works with formal power series in the
parameter $\epsilon$, and the construction is in fact an infinite
recurrence. In particular,
$\mathcal{K}(\psi,J,\epsilon)=\mathcal{K}_0(\psi,J)
+\epsilon\mathcal{K}_1(\psi,J)+\ldots$\,,
where the leading term remains untouched,
\begin{equation}
  \label{eq:calK_0}
  \mathcal{K}_0(\psi,J) = K_0(\psi,J) = \omega_1J_1+\omega_2J_2.
\end{equation}
In practice one has to interrupt the recurrence at some order which
means replacing the true Hamiltonian system by an approximate averaged
system. In our case we shall be content with a truncation at the first
order.

In our model, one can apply the substitution
$\chi_1=\mu\psi_1-\nu\psi_2$,
$\chi_2=\tilde{\nu}\psi_1+\tilde{\mu}\psi_2$ where
$\tilde{\mu},\tilde{\nu}\in\mathbb{Z}$ are such that
$\tilde{\mu}\mu+\tilde{\nu}\nu=1$ ($\tilde{\mu},\tilde{\nu}$ exist
since $\mu$, $\nu$ are coprime). The phase $\chi_2$ is classified as
nonresonant and can be eliminated while the phase $\chi_1$ is resonant
and survives the canonical transformation. The procedure is explained
in more detail in Subsections~\ref{ses:von_Zeipel_dynamics_1storder}
and \ref{sec:1st_order_dynamics}. Recalling (\ref{eq:calK_0}) combined
with (\ref{eq:w2_to_w1_rational}), here we just give the resulting
formula for the averaged Hamiltonian truncated at the first order,
\begin{equation}
  \label{eq:K_trunc_1}
  \mathcal{K}_{(1)}(\psi,J,\epsilon)
  = \frac{\omega_{1}}{\nu}\left(\nu J_{1}+\mu J_{2}\right)
  +\epsilon\,\mathcal{K}_{1}(\psi,J)
\end{equation}
where
\begin{equation}
  \mathcal{K}_{1}(\psi,J)
  = -\frac{\omega_{1}}{2}\sum_{n\in\mathbb{Z}\setminus\{0\}}
  \sF[f]_{-n\nu}\, i^{n\mu}\beta(J_{1})^{|n|\mu}\,
  e^{in(\mu\psi_{1}-\nu\psi_{2})}
  \label{eq:K1cal_psi_J}
\end{equation}
and
\begin{equation}
  \beta(J_{1})=\sqrt{\frac{J_{1}}{J_{1}+p_{\theta}}}\,.
  \label{eq:def_beta}
\end{equation}

One observes that the averaged four-dimensional Hamiltonian system
admits a reduction to a two-dimensional subsystem. Indeed, as it
should be, $\mathcal{K}_{1}(\psi,J)$ depends on the angles only
through the combination $\mu\psi_{1}-\nu\psi_{2}=\chi_{1}$. It follows
that the action $\nu J_{1}+\mu J_{2}$ is an integral of motion for the
Hamiltonian $\mathcal{K}_{(1)}(\psi,J,\epsilon)$. The reduced system
then depends on the coordinates $\chi_1$, $J_1$.

Moreover, by inspection of the series (\ref{eq:K1cal_psi_J}) one finds
that the Hamiltonian for the reduced subsystem can be expressed in
terms of a single complex variable $z=\beta(J_1)^\mu{}e^{i\chi_1}$ and
takes the form $\mathcal{Z}(\chi_{1},J_{1})=\Re[h(z)]$, with $h$ being
a holomorphic function. We refer to
Subsection~\ref{sec:1st_order_dynamics} for the proof of the following
abstract theorem.

\begin{theorem}
  \label{thm:AsymptHamDynB1}
  Let $h\in C^{1}(\overline{B_{1}})$ and suppose $h(z)$ is a
  nonconstant holomorphic function on the open unit disk $B_{1}$. Let
  $\varrho:[0,+\infty[\:\to[0,1[$ be a smooth function such that
  $\varrho'(x)>0$ for $x>0$, $\varrho(0)=0$ and
  $\lim_{x\to+\infty}\varrho(x)=1$. Let $\mathcal{Z}(\chi_{1},J_{1})$
  be the Hamilton function on $\mathbb{R}\times\,]0,+\infty[$ defined
  by
  \[
  \mathcal{Z}(\chi_{1},J_{1})
  = \Re\!\left[h\!\left(\varrho(J_{1})e^{i\chi_{1}}\right)\right].
  \]
  Then, for almost all initial conditions $(\chi_{1}(0),J_{1}(0))$,
  the corresponding Hamiltonian trajectory fulfills
  \begin{equation}
    \lim_{t\to+\infty}\chi_{1}(t)=\chi_{1}(\infty)\in\mathbb{R},
    \textrm{~}\lim_{t\to+\infty}J_{1}(t)=+\infty,
    \label{eq:limInfChiJ}
  \end{equation}
  and
  \begin{equation}
    \lim_{t\to+\infty}J_{1}'(t)
    = \Im\!\left[e^{i\chi_{1}(\infty)}\,
      h'\!\left(e^{i\chi_{1}(\infty)}\right)\right]>0.
    \label{eq:limInfDerChiJ}
  \end{equation}
\end{theorem}

Theorem~\ref{thm:AsymptHamDynB1} yields the desired information about
the asymptotic behavior of the averaged Hamiltonian system in
action-angle coordinates $(J,\psi)$. One has to go back to the
original action-angle coordinates $(I,\varphi)$, however, since in
these coordinates the dynamics of the studied system can be directly
interpreted. This means to invert the canonical transformation
$(I,\varphi)\mapsto(J,\psi)$ which resulted from the
Poincar\'e-von~Zeipel elimination method. Let us remark that the
generating function of this canonical transformation is truncated at
the first order as well. Doing so one derives the following result
whose proof can again be found in
Subsection~\ref{sec:1st_order_dynamics}.

\begin{proposition}
  \label{thm:ConclusionResonant}
  Suppose that conditions (\ref{eq:w2_to_w1_rational}),
  (\ref{eq:ResonantCase}) characterizing the resonant case are
  satisfied. Let $(\varphi_1(t),I_1(t))$ be a trajectory for the
  first-order averaged Hamiltonian (\ref{eq:K_trunc_1}) but expressed
  in the original angle-action coordinates, as introduced in
  Subsection~\ref{results_Hamiltonian}.  Then, for almost all initial
  conditions $(\varphi_{1}(0),I_{1}(0))$,
  \begin{equation}
    \lim_{t\to+\infty}\left(\varphi_{1}(t)-\omega_{1}t\right)
    = \phi(\infty)\in\mathbb{R},\textrm{~~}
    \lim_{t\to+\infty}\frac{I_{1}(t)}{t}=C>0,
    \label{eq:vonZeipel_asympt_domain}
  \end{equation}
  and
  \[
  C = -\frac{\varepsilon\omega_{2}}{2}\,
  f_{\nu}'\!\left(-\left(\phi(\infty)+\frac{\pi}{2}\right)
    \lambda\right),\ \mbox{with}\ f_{\nu}(\varphi)
  = \langle f(\varphi)\rangle_{\mathbb{Z}\nu}.
  \]
\end{proposition}

\subsection{A decoupled system and the leading asymptotic terms}
\label{sec:results_asymptotic}

In the preceding subsection we studied an approximate dynamics derived
with the aid of averaging methods. It turned out that a typical
trajectory reaches for large times an asymptotic domain characterized
by formulas (\ref{eq:vonZeipel_asympt_domain}). Being guided by this
experience as well as by numerical experiments that we have carried
out we suggest that formulas (\ref{eq:vonZeipel_asympt_domain}) in
fact give the correct leading asymptotic terms for the complete
system.

To support this suggestion we now derive the leading asymptotic terms
and a bound on the order of error terms for the original Hamiltonian
equations (\ref{eqmotion1}), (\ref{eqmotion2}) while assuming that the
system already follows the anticipated asymptotic regime. At the same
time, we relax the assumption on smallness of the parameter
$\epsilon$.

Applying the substitution
\[
F(t)=2I(t)+|a(t)|,\mbox{~}\phi(t)=\varphi(t)-bt,
\]
to equations (\ref{eqmotion1}), (\ref{eqmotion2}) we obtain the system
of differential equations
\begin{eqnarray}
  F'(t) & = & \frac{a(t)\, a'(t)}{F(t)+\sqrt{F(t)^{2}-a(t)^{2}}
    \sin(bt+\phi(t))}\,,\label{eq:ODEFDevelopEps}\\
  \noalign{\medskip}
  \phi'(t) & = & -\frac{a(t)\, a'(t)
    \cos(bt+\phi(t))}{\sqrt{F(t)^{2}-a(t)^{2}}
    \left(F(t)+\sqrt{F(t)^{2}-a(t)^{2}}\sin(bt+\phi(t))\right)}\,,
  \label{eq:ODEphiDevelopEps}
\end{eqnarray}
where $a(t)=p_{\theta}-\epsilon f(\Omega t)$ (see (\ref{eq:DefAt}) and
(\ref{eq:Phi_rel_f})). The real function $f(t)$ is supposed to be
continuously differentiable and $2\pi$-periodic. The only assumptions
imposed on $\epsilon>0$ in this subsection are that $\epsilon$ is not
too big so that the function $a(t)$ has no zeroes and is everywhere of
the same sign as $p_{\theta}$. Recall that for definiteness
$p_{\theta}$ is supposed to be positive. Clearly, the functions $a(t)$
and $a'(t)$ are bounded on $\mathbb{R}$.

Equations (\ref{eq:ODEFDevelopEps}) and (\ref{eq:ODEphiDevelopEps})
are nonlinear and coupled together. To decouple them we replace
$\phi(t)$ on the RHS of (\ref{eq:ODEFDevelopEps}) and $F(t)$ on the
RHS of (\ref{eq:ODEphiDevelopEps}) by the respective leading
asymptotic terms (\ref{eq:vonZeipel_asympt_domain}), as learned from
the averaging method. This is done under the assumption that the
solution has already reached the domain $F(t)\geq F_{0}\gg p_{\theta}$
where $F(t)$ is sufficiently large and starts to grow. Thus, to
formulate a problem with decoupled equations, we replace $\phi(t)$ in
(\ref{eq:ODEFDevelopEps}) by the expected limit value
$\phi\equiv\phi(\infty)\in\mathbb{R}$, i.e. the simplified equation
reads
\begin{equation}
  F'(t)=\frac{a(t)\, a'(t)}{F(t)+\sqrt{F(t)^{2}-a(t)^{2}}\,
    \sin(bt+\phi)}\,.
  \label{eq:ODEsimplifiedF}
\end{equation}

As stated in Proposition~\ref{thm:F} below, a solution $F(t)$ of
(\ref{eq:ODEsimplifiedF}) actually grows linearly for large times.
Conversely, equation (\ref{eq:ODEphiDevelopEps}) is analyzed in
Proposition~\ref{thm:asym_phit_F_lin} under the assumption that $F(t)$
grows linearly. In that case $\phi(t)$ is actually shown to approach a
constant value as $t$ tends to infinity. For derivation of
Proposition~\ref{thm:asym_phit_F_lin} the periodicity of the functions
$a(t)$ is not important. It suffices to assume that it takes values
from a bounded interval separated from zero. For the proofs see
Section~\ref{sec:SimplifiedEquation}.

\begin{proposition}
  \label{thm:F}
  Suppose $f\in C^{2}(\mathbb{R})$ is $2\pi$-periodic and
  \[
  \frac{\Omega}{b}\in\mathbb{N},\mbox{~}
  f'\!\left(-\frac{\Omega}{b}\left(\phi+\frac{\pi}{2}\right)\right)<0.
  \]
  Then any solution of (\ref{eq:ODEsimplifiedF}) such that $F(0)$ is
  sufficiently large fulfills
  \[
  F(t)=\epsilon\,\Omega\left|f'\!\left(-\frac{\Omega}{b}
      \left(\phi+\frac{\pi}{2}\right)\right)\right|\, t
  +O\!\left(\log(t)^{2}\right)\textrm{~~as~}t\to+\infty.
  \]
\end{proposition}

\begin{proposition}
  \label{thm:asym_phit_F_lin}
  Suppose $a(t)\in C^{1}(\mathbb{R})$ fulfills
  \[
  A_{1}\leq a(t)\leq A_{2},\mbox{\ }|a'(t)|\leq A_{3},
  \]
  for some positive constants $A_{1}$, $A_{2}$, $A_{3}$. Furthermore,
  suppose $F(t)\in C(\mathbb{R})$ has the asymptotic behavior
  \begin{equation}
    F(t)=\alpha t+O\!\left(\log(t)^{2}\right)\quad\mbox{as\ }t\to+\infty,
    \label{eq:Ft_asympt}
  \end{equation}
  with a positive constant $\alpha$. If $\phi(t)$ obeys the
  differential equation (\ref{eq:ODEphiDevelopEps}) on a neighborhood
  of $+\infty$, then there exists a finite limit
  $\lim_{t\to+\infty}\phi(t)=\phi(\infty)$ and
  \begin{equation}
    \phi(t)=\phi(\infty)+O\!\left(\frac{\log(t)}{t}\right)
    \quad\mbox{as\ }t\to+\infty.
    \label{eq:phit_asympt}
  \end{equation}
\end{proposition}

{\scshape Remarks.} (i)\ Let us point out an essential difference
between equations (\ref{eq:ODEFDevelopEps}) and
(\ref{eq:ODEphiDevelopEps}). Note that for all $F,a,s\in\mathbb{R}$
such that $F\geq|a|>0$ one has
\begin{equation}
  \left|\frac{a\cos(s)}{F+\sqrt{F^{2}-a^{2}}\sin(s)}\right| \leq 1,
  \label{eq:aux1}
\end{equation}
and, consequently, it follows from (\ref{eq:ODEphiDevelopEps}) that
\begin{equation}
  |\phi'(t)|\leq\frac{|a'(t)|}{\sqrt{F(t)^{2}-a(t)^{2}}}\,.
  \label{eq:phi_der_estim}
\end{equation}
Thus the RHS of (\ref{eq:ODEphiDevelopEps}) is inversely proportional
to $F(t)$. The point is that if the denominator in (\ref{eq:aux1})
becomes very small for those $s$ for which $\sin(s)=-1$ then one gets
a compensation by the vanishing numerator. Nothing similar can be
claimed, however, for
equation (\ref{eq:ODEFDevelopEps}).\\
(ii)\ The replacement of the phase $\phi(t)$ by a constant in
(\ref{eq:ODEFDevelopEps}) was quite crucial for derivation of the
result stated in Proposition~\ref{thm:F}. In fact, suppose that $F(t)$
is sufficiently large. In the case of the original equation
(\ref{eq:ODEFDevelopEps}), too, essential contributions to the
increase of $F(t)$ are achieved at the moments of time for which
$\sin(bt+\phi(t))=-1$. If $\phi(t)$ equals a constant then these
moments of time are well defined and the growth of $F(t)$ can be
estimated. On the contrary, without a sufficiently precise information
about $\phi(t)$ one loses any control on the growth of $F(t)$.
\smallskip

Propositions~\ref{thm:F} and \ref{thm:asym_phit_F_lin} can also be
interpreted in the following way. Let us pass from the differential
equations (\ref{eq:ODEFDevelopEps}), (\ref{eq:ODEphiDevelopEps}) to
the integral equations
\begin{eqnarray*}
  &  & \hspace{-0.3em}F(t)-F(0)
  -\int_{0}^{t}\frac{a(s)\, a'(s)}{F(s)
    +\sqrt{F(s)^{2}-a(s)^{2}}\,\sin(bs+\phi(s))}\,\mbox{d}s=0,\\
  \noalign{\medskip} &  & \hspace{-0.3em}\phi(t)-\phi(\infty)
  -\int_{t}^{\infty}\frac{a(s)\, a'(s)\cos(bs+\phi(s))}
  {\sqrt{F(s)^{2}-a(s)^{2}}
    \left(F(s)+\sqrt{F(s)^{2}-a(s)^{2}}\sin(bs+\phi(s))\right)}\,
  \mbox{d}s = 0.
\end{eqnarray*}
Suppose $\phi(\infty)$ satisfies
$\alpha=-\epsilon\,\Omega{}f'(-(\Omega/b)(\phi(\infty)+\pi/2))>0$. If
$F(0)$ is sufficiently large then the functions
\[
F(t)=\alpha t+F(0),\ \phi(t)=\phi(\infty),\ t>0,
\]
can be regarded as an approximate solution of this system of integral
equations with errors of order $O(\log(t)^{2})$ for the first equation
and of order $O(\log(t)/t)$ for the second one.

One has to admit, however, that this argument still does not represent
a complete mathematical proof of the asymptotic behavior of the
action-angle variables $I(t)=(F(t)-|a(t)|)/2$,
$\varphi(t)=bt+\phi(t)$. So far we have analyzed either the dynamics
generated by an approximate averaged Hamiltonian in
Subsection~\ref{sec:results_averaged} or a simplified decoupled system
in the current subsection. Moreover, in the latter case it should be
emphasized that the simplified equations were derived under the
essential assumption that the dynamical system had already reached the
regime characterized by an acceleration with an unlimited energy
growth (this is reflected by the assumption that $F(0)$ is
sufficiently large). Nevertheless, on the basis of this analysis as
well as on the basis of numerical experiments we formulate the
following conjecture.

\begin{conjecture}
  \label{thm:Conj_asympt}
  If $\Omega/b\in\mathbb{N}$ then the regime of acceleration for the
  original (true) dynamical system is described by the asymptotic
  behavior
  \begin{equation}
    I(t) = Ct+O\!\left(\log(t)^{2}\right)\!,\
    \varphi(t)=bt+\phi(\infty)+O\!\left(\frac{\log(t)}{t}\right)\ \
    \mbox{as\ }t\to+\infty,
    \label{eq:I_vphi_asympt}
  \end{equation}
  where $\phi(\infty)$ is a real constant and
  \begin{equation}
    C = -\frac{1}{2}\,\epsilon\Omega
    f'\!\left(-\xi\right)>0,\,\mbox{with\ \ }
    \xi = \frac{\Omega}{b}\left(\phi(\infty)+\frac{\pi}{2}\right).
    \label{eq:C_xi_asympt_def}
  \end{equation}
\end{conjecture}

\subsection{Guiding center coordinates}
\label{sec:results_guiding}

Being given the asymptotic relations (\ref{eq:I_vphi_asympt}),
(\ref{eq:C_xi_asympt_def}) it is desirable to describe the accelerated
motion in terms of the original Cartesian coordinates $q$. The
description becomes more transparent if the motion is decomposed into
a motion of the guiding center $X$ and a relative motion of the
particle with respect to this center which is characterized by a
gyroradius vector $R$ and a gyrophase $\vartheta$ \cite{Littlejohn}.
Since the presented results have a direct physical interpretation, in
this subsection we make an exception and give the formulas using now
all physical constants (including $m$ and $e$). Additional details and
derivations are postponed to Section~\ref{sec:Conclusion}.

Let $v(q,p,t)=(p-eA(q,t))/m$ be the velocity. We write $q=X+R$ where,
by definition,
\begin{displaymath}
  X(q,p,t) = q+\frac{1}{\omega_c}\, v^{\perp}(q,p,t),\ 
  R(q,p,t) = q-X(q,p,t) = -\frac{1}{\omega_c}\, v^{\perp}(q,p,t)
\end{displaymath}
are the guiding center field and the gyroradius field, respectively.
In what follows, we use the polar decompositions
\begin{displaymath}
  q = r\,(\cos\theta,\sin\theta),\ X = |X|(\cos\chi,\sin\chi),\
  R = |R|(\cos\vartheta,\sin\vartheta).
\end{displaymath}
Concerning the geometrical arrangement, one has the relation
\begin{equation}
  \vartheta=\varphi+\chi-\frac{\pi}{2}\ \ (\mbox{mod}\ 2\pi)
  \label{eq:varth_rel_varphi_chi}
\end{equation}
(with $\varphi$ being introduced in
Subsection~\ref{results_Hamiltonian}, for the derivation see
Section~\ref{sec:Conclusion}).

The quantities $X$, $R$ were introduced (under different names) and
studied in \cite{AschStovicek} where one can also find several
formulas given below, notably those given in
(\ref{eq:centervelocity}). Observe that
$|p|^{2}=p_{r}^{\,2}+p_{\theta}^{\,2}/r^{2}$ and, according to
(\ref{eq:DefHradial}) and (\ref{eq:I_fce_polar}),
\begin{equation}
  \label{eq:rel_H_I}
  H = H_{\text{rad}}+\frac{\omega_ca}{2}
  = \omega_c\!\left(I+\frac{1}{2}\,(|a|+a)\right),
  \text{~~with~}
  a(t) = p_\theta-\frac{e}{2\pi}\,\Phi(t).
\end{equation}
Using these relations one derives that
\begin{equation}
  \frac{m\omega_{c}^{\,2}}{2}\,|X|^{2}
  =  H-\omega_{c}\, p_{\theta}+\frac{e\omega_{c}}{2\pi}\,\Phi,~~
  \frac{m\omega_{c}^{\,2}}{2}\,|R|^{2} = H.
  \label{eq:centervelocity}
\end{equation}
Observe from (\ref{eq:centervelocity}) that if $a(t)$ is an everywhere
positive function then $|R(t)|>|X(t)|$ and so the center of
coordinates always stays in the domain encircled by the spiral-like
trajectory.  On the contrary, if $a(t)$ is an everywhere negative
function then the center of coordinates is never encircled by the
trajectory.

For a Hamiltonian trajectory $(q(t),p(t))$ put $H(t)=H(q(t),p(t),t)$.
Suppose again that $p_\theta>0$ and hence $a(t)>0$. Applying
Conjecture~\ref{thm:Conj_asympt} and recalling (\ref{eq:Phi_rel_f})
one has
\begin{equation}
  \label{eq:I_of_t_asympt}
  I(t) = -\frac{e}{4\pi}\,
  \Phi'\!\left(-\frac{1}{\omega_{c}}\left(\phi(\infty)
      +\frac{\pi}{2}\right)\right)t
  +O\!\left(\log(t)^2\right)
  \text{~~as~}t\to+\infty.
\end{equation}
Using (\ref{eq:rel_H_I}) and recalling definition (\ref{eq:rate}) of
the acceleration rate $\gamma$ one finds the positive value
\begin{equation}
  \label{eq:gamma}
  \gamma = \lim_{\tau\to\infty}\frac{H(\tau)}{\tau}
  = -\frac{e\omega_{c}}{4\pi}\,
  \Phi'\!\left(-\frac{1}{\omega_{c}}\left(\phi(\infty)
      +\frac{\pi}{2}\right)\right)\! > 0.
\end{equation}
From (\ref{eq:centervelocity}), (\ref{eq:rel_H_I}) and
(\ref{eq:I_of_t_asympt}) one also deduces the asymptotic behavior of
the guiding center and the gyroradius,
\begin{equation}
  |X(t)| = \sqrt{\frac{2\gamma t}{m\omega_c^{\,2}}}\,
  +O\!\left(\frac{\log(t)^{2}}{\sqrt{t}}\right)\!,\
  |R(t)| = \sqrt{\frac{2\gamma t}{m\omega_c^{\,2}}}\,
  +O\!\left(\frac{\log(t)^{2}}{\sqrt{t}}\right)\ \ \mbox{as\ }
  t\to+\infty.
  \label{eq:XR_norm_asympt}
\end{equation}
Now it is obvious why the resonance has also consequences for the
drift. The growing energy $H$ comes along with a growing distance of
the guiding center from the origin. Thus in $q$-space the particle
oscillates around a drifting center which goes out at the same rate as
the radius grows.

From Conjecture~\ref{thm:Conj_asympt}, with a bit of additional
analysis (see Section~\ref{sec:Conclusion}), one can also derive
consequences for the asymptotic behavior of the angle variables. One
has, as $t\to+\infty$,
\begin{equation}
  \label{eq:chi_vtheta_asympt}
  \chi(t) = D\log(\omega_ct)+\chi(\infty)+o(1),~~
  \vartheta(t) = \omega_ct+D\log(\omega_ct)+\vartheta(\infty)+o(1),
\end{equation}
where $D$, $\chi(\infty)$ and $\vartheta(\infty)$ are real constants.
Relations (\ref{eq:XR_norm_asympt}) and (\ref{eq:chi_vtheta_asympt})
give a complete information about the asymptotic behavior of the
trajectory $q(t)=X(t)+R(t)$.

Finally, in connection with formula (\ref{eq:gamma}) let us point out
that the studied system behaves for large times almost as a ``kicked''
system. Note that, for a fixed $a>0$,
\begin{displaymath}
  \frac{a}{x-\sqrt{x^{2}-a^{2}}\,\cos(t)}
  = \sum_{n=-\infty}^{+\infty}
  \left(\frac{x-a}{x+a}\right)^{\!|n|/2}e^{int}
  \,\to\, 2\pi\,\delta_{2\pi}(t)
  \text{~~as~}x\to+\infty,
\end{displaymath}
where $\delta_A(t)=\sum_{n=-\infty}^{+\infty}\delta(t-nA)$ stands for
the $A$-periodic prolongation of the Dirac $\delta$ function. Since
$H'(t)=-e\theta'(t)\Phi'(t)/(2\pi)$ (see (\ref{eq:D_der})),
$\varphi(t)\approx\omega_ct+\phi(\infty)$ and, in view of
(\ref{eq:canonic_r}),
\begin{displaymath}
  \theta' = \frac{a}{mr^2}+\frac{\omega_c}{2}
  = \frac{\omega_{c}a}{2\!\left(2I+a+2\sqrt{I(I+a)}\,
      \sin\varphi\right)}+\frac{\omega_c}{2}
\end{displaymath}
one has
\begin{displaymath}
  H'(t) \approx -\frac{e}{2}\,
  \Phi'\!\left(-\frac{1}{\omega_{c}}\left(\phi(\infty)
      +\frac{\pi}{2}\right)\right)
  \delta_{2\pi/\omega_c}\!\left(t+
    \frac{1}{\omega_c}\!\left(\phi(\infty)+\frac{\pi}{2}\right)\right).
\end{displaymath}
Thus the system gains the main contributions to the energy growth in
narrow intervals around the instances of time for which
$\omega_ct=-\phi(\infty)-\pi/2$ (mod~$2\pi$). According to
(\ref{eq:varth_rel_varphi_chi}), these are exactly those instances of
time when $\vartheta(t)-\chi(t)\approx\pi$ (mod~$2\pi$), i.e. when the
particle passes next to the singular flux line.

\section{Derivation of the Hamiltonian in action-angle coordinates}
\label{sec:TheModel}

In this section we derive the transformation equations from
coordinates $r$, $p_r$ to action-angle coordinates $I$, $\varphi$.

Let us first note that the equations of motion for the radial
Hamiltonian (\ref{eq:DefHradial}) have a well-defined unique solution
for any initial condition and for all times. They are equivalent to
the nonlinear second-order differential equation
\begin{equation}
  r''+\frac{b^{2}}{4}\,r = \frac{a(t)^{2}}{r^{3}}\,.
  \label{eq:r_nonlin_2order}
\end{equation}

\begin{proposition}
  \label{thm:complete_rt}
  Suppose $a(t)$ is a real continuously differentiable function
  defined on $\mathbb{R}$ having no zeros. Then for any initial
  condition $r(t_{0})=r_{0}$, $r'(t_{0})=r_{1}$, with
  $(t_{0},r_{0},r_{1})\in\mathbb{R}\times\,]0,+\infty[\,\mathbb{\times{}R}$,
  there exists a unique solution of the differential equation
  (\ref{eq:r_nonlin_2order}) defined on the whole real line
  $\mathbb{R}$ and satisfying this initial condition.
\end{proposition}

\begin{proof}
  Suppose $r(t)$ is a solution of the Hamiltonian equations
  (\ref{eq:r_nonlin_2order}). Put $p_r=r'$ and
  $H(t)=H_{\mathrm{rad}}(r(t),p_{r}(t),t)$. Then
  \[
  \left|\frac{\mbox{d}}{\mbox{d}t}H(t)\right|
  =\frac{|a(t)a'(t)|}{r(t)^{2}}\leq\left|\frac{2a'(t)}{a(t)}\right|H(t).
  \]
  From here one readily concludes that if $r(t)$ is defined on a
  bounded interval $M\subset\mathbb{R}$, then there exist constants
  $R_{1}$, $R_{2}$, $0<R_{1}\leq R_{2}<+\infty$, such that $R_{1}\leq
  r(t)\leq R_{2}$ for all $t\in M$. From the general theory of
  ordinary differential equations it follows that any solution $r(t)$
  of (\ref{eq:r_nonlin_2order}) can be continued to the whole real
  line.
\end{proof}

As a first step of the derivation of the transformation we introduce
the action-angle coordinates for a frozen time. Assume for a moment
that $a(t)=a$ is a constant and denote
\[
V(r)=\frac{a^{2}}{2r^{2}}+\frac{b^{2}r^{2}}{8}.
\]
Suppose a fixed energy level $E$ is greater than the minimal value
$V_{\textrm{min}}=b|a|/2$. Then the motion is constrained to a bounded
interval $[r_{-},r_{+}]$, and one has
\[
E-V(r)=\frac{b^{2}}{8r^{2}}\,(r_{+}^{\;2}-r^{2})(r^{2}-r_{-}^{\;2}).
\]
The action equals
\[
I(E)=\frac{1}{\pi}\int_{r_{-}}^{r_{+}}\sqrt{2(E-V(\rho))}\,
\mathrm{d}\rho=\frac{b}{8}\,(r_{+}-r_{-})^{2}.
\]
Hence
\begin{displaymath}
  r_{\pm}=\frac{2}{\sqrt{b}}
  \left(I+\frac{|a|}{2}\pm\sqrt{I(I+|a|)}\right)^{\!1/2}
  =\sqrt{\frac{2}{b}}\left(\sqrt{I+|a|}\pm\sqrt{I}\right).
\end{displaymath}

Using the generating function,
\[
S(r,I)=\int_{r_{-}}^{r}\sqrt{2(E-V(\rho))}\,\mathrm{d}\rho
=\frac{b}{2}\int_{r_{-}}^{r}\frac{1}{\rho}\,
\sqrt{(r_{+}^{\;2}-\rho^{2})(\rho^{2}-r_{-}^{\;2})}\,\mathrm{d}\rho,
\]
one derives the canonical transformation of variables between
$(r,p_{r})$ and the angle-action variables $(\varphi,I)$. One has
\begin{eqnarray*}
  \frac{\partial S}{\partial I}
  & = & \frac{b}{4}\int_{r_{-}}^{r}\frac{1}{\rho}
  \left(\sqrt{\frac{\rho^{2}-r_{-}^{\,2}}{r_{+}^{\,2}-\rho^{2}}}\,
    \frac{\mbox{d}r_{+}(I)^{2}}{\mbox{d}I}
    -\sqrt{\frac{r_{+}^{\,2}-\rho^{2}}{\rho^{2}-r_{-}^{\,2}}}\,
    \frac{\mbox{d}r_{-}(I)^{2}}{\mbox{d}I}\right)\mathrm{d}\rho\\
  & = & \frac{\pi}{2}-\arctan\!\left(\frac{r_{+}^{\,2}
      +r_{-}^{\,2}-2r^{2}}{2\sqrt{(r_{+}^{\,2}-r^{2})(r^{2}-r_{-}^{\,2})}}
  \right)\!.
\end{eqnarray*}
For the angle variable $\varphi=\partial S/\partial I-\pi/2$ one
obtains
\[
\sin(\varphi)=\frac{1}{\sqrt{I(I+|a|)}}
\left(\frac{br^{2}}{4}-I-\frac{|a|}{2}\right)\!.
\]
Furthermore,
\[
p_{r}=\frac{\partial S}{\partial r}
=\frac{b}{2r}\,\sqrt{(r_{+}^{\,2}-r^{2})(r^{2}-r_{-}^{\,2})}
=\frac{2}{r}\,\sqrt{I(I+|a|)}\,\cos(\varphi).
\]
Relations (\ref{eq:canonic_r}), (\ref{eq:canonic_pr}),
(\ref{eq:I_fce_polar}), (\ref{eq:varphi_fce_polar}) readily follow.

Let us now switch to the time-dependent case. The Hamiltonian
transforms according to the rule
\[
H_{c}(\varphi,I,t)=H_{\mathrm{rad}}\big(r(\varphi,I,t),
p_{r}(\varphi,I,t),t\big)
+\frac{\partial S(u,I,t)}{\partial t}\Bigg\vert_{u=r(\varphi,I,t)}.
\]
One computes
\begin{displaymath}
  \frac{\partial S(u,I,t)}{\partial t}\Bigg\vert_{u=r(\varphi,I,t)}
  = \frac{b\,|a|'}{4}\int_{r_{-}}^{r(\varphi,I,t)}
  \frac{1}{\rho}\left(\sqrt{
      \frac{\rho^{2}-r_{-}^{\,2}}{r_{+}^{\,2}-\rho^{2}}}\,
    \frac{\partial r_{+}^{\,2}}{\partial|a|}
    -\sqrt{\frac{r_{+}^{\,2}-\rho^{2}}{\rho^{2}-r_{-}^{\,2}}}\,
    \frac{\partial r_{-}^{\,2}}{\partial|a|}\right)\mathrm{d}\rho.
\end{displaymath}
Simplifying the expression and dropping those terms which are
independent of $\varphi$ and $I$ one finally arrives at the expression
(\ref{eq:H_c_def}).

\section{Details on the von Zeipel elimination method and the proofs}
\label{sec:vonZeipelMethod}

\subsection{A summary of basic formulas}
\label{sec:vonZeipel_Notation}

This is a short summary of basic steps in the Poincar\'e-von~Zeipel
elimination procedure. Recall the notation from the beginning of
Subsection~\ref{sec:results_averaged}.

Consider a completely integrable Hamiltonian in action-angle
coordinates, $K_{0}(I)=\omega\cdot I$, where $I$ runs over a domain in
$\mathbb{R}^{d}$, $\varphi\in\mathbb{T}^{d}$ and
$\omega\in\mathbb{R}_{+}^{\, d}$ is a constant vector of frequencies.
One is interested in a perturbed system with a small Hamiltonian
perturbation so that the total Hamiltonian reads
\[
K(\varphi,I,\epsilon)
=K_{0}(I)+\epsilon\, K_{*}(\varphi,I,\epsilon)
=K_{0}(I)+\epsilon\, K_{1}(\varphi,I)+\epsilon^{2}\,
K_{2}(\varphi,I)+\ldots
\]
where $\epsilon$ is a small parameter. The function
$K_{*}(\varphi,I,\epsilon)$ is assumed to be analytic in all
variables

Let $\mathbb{K}$ be the lattice of indices in $\mathbb{Z}^{d}$
corresponding to resonant frequencies and $\mathbb{K}^{c}$ be its
complement, i.e.
\begin{displaymath}
  \mathbb{K} = \{n\in\mathbb{Z}^{d};\,n\cdot\omega=0\},\mbox{~}
  \mathbb{K}^{c}=\mathbb{Z}^{d}\setminus\mathbb{K}.
\end{displaymath}
One applies a formal canonical transformation of variables,
$(I,\varphi)\mapsto(J,\psi)$, so that the Fourier series in the angle
variables $\psi$ of the resulting Hamiltonian
$\mathcal{K}(\psi,J,\epsilon)$ has nonzero coefficients only for
indices from the lattice $\mathbb{K}$. The canonical transformation is
generated by a function $S(\varphi,J,\epsilon)$ regarded as a formal
power series with coefficient functions $S_{j}(\varphi,J)$ and the
absolute term $S_{0}(\varphi,J)=\varphi\cdot J$. Similarly, the new
Hamiltonian $\mathcal{K}(\psi,J,\varepsilon)$ is sought in the form of
a formal power series with coefficient functions
$\mathcal{K}_{j}(\psi,J)$. One arrives at the system of equations
$\mathcal{K}_{0}(J,\varphi)=K_{0}(J)=\omega\cdot J$ and
\[
\mathcal{K}_{j}(\varphi,J)
=\omega\cdot\frac{\partial S_{j}(\varphi,J)}{\partial\varphi}
+P_{j}(\varphi,J),\ j\geq1,
\]
where $P_{1}(\varphi,J)=K_{1}(\varphi,J)$ and the terms $P_{j}$ for
$j\geq2$ are determined recursively. The formal von Zeipel Hamiltonian
is defined by the equalities
$\mathcal{K}_{j}(\psi,J)=\langle{}P_{j}(\psi,J)\rangle_{\mathbb{K}}$
for $j\geq1$. Coefficients $S_{j}(\varphi,J)$ are then solutions of
the first order differential equations
\[
\omega\cdot\frac{\partial S_{j}(\varphi,J)}{\partial\varphi}
=-\left\langle P_{j}(\varphi,J)\right\rangle _{\mathbb{K}^{c}},
\mbox{~}j\geq1.
\]

In practice one truncates $\mathcal{K}(\psi,J,\epsilon)$ at some order
$m\geq1$ of the parameter $\epsilon$. Let us define the $m$th order
averaged Hamiltonian
\[
\mathcal{K}_{(m)}(\psi,J,\epsilon)
=\mathcal{K}_{0}(J)+\epsilon\,\mathcal{K}_{1}(\psi,J)
+\ldots+\epsilon^{m}\mathcal{K}_{m}(\psi,J).
\]
Similarly, let $S_{(m)}(\varphi,J,\epsilon)$ be the truncated
generating function. If $(\psi(t),J(t))$ is a solution of the
Hamiltonian equations for $\mathcal{K}_{(m)}(\psi,J,\epsilon)$, and if
$(\varphi(t),I(t))$ is the same solution after the inverted canonical
transformation generated by $S_{(m)}(\varphi,J,\epsilon)$, then
$(\varphi(t),I(t))$ is expected to approximate well the solution of
the original system (governed by the Hamiltonian
$K(\varphi,I,\epsilon)$) for times of order $1/\epsilon^{m}$ (see
\cite{ArnoldKozlovNeishtadt} for a detailed discussion).

\subsection{Derivation of the first-order averaged Hamiltonian}
\label{ses:von_Zeipel_dynamics_1storder}

If the ratio $\omega_{2}/\omega_{1}$ is irrational then the lattice
$\mathbb{K}$ is trivial, $\mathbb{K}=\{0\}$, and the von Zeipel method
amounts to the ordinary averaging method in the angle variables
$\varphi$. Here we focus on the complementary case
(\ref{eq:w2_to_w1_rational}) when $\omega_{2}/\omega_{1}=\mu/\nu$,
with $\mu$ and $\nu$ being coprime positive integers.

Recalling (\ref{eq:K_extended}), (\ref{eq:F_vphi_I_eps}), we have
$K(\varphi,I,\epsilon)=\omega_{1}I_{1}+\omega_{2}I_{2}+\epsilon
K_{1}(\varphi,I)+\epsilon^{2}\widetilde{K}(\varphi,I,\epsilon)$ where
$\widetilde{K}(\varphi,I,\epsilon)$ is an analytic function in
$\epsilon$,
\[
K_{1}(\varphi,I)=\omega_{2}f'(\varphi_{2})
F_{1}(\varphi_{1},I_{1}),\
F_{1}(\varphi_{1},I_{1})
=\arctan\!\left(\frac{\sqrt{I_{1}}
    \cos(\varphi_{1})}{\sqrt{I_{1}+p_{\theta}}+\sqrt{I_{1}}
    \sin(\varphi_{1})}\right)\!.
\]
One finds that
\[
\sF[F_{1}(\varphi_{1},I_{1})]_{k}=\frac{i^{k-1}}{2k}
\left(\frac{I_{1}}{I_{1}+p_{\theta}}\right)^{\!|k|/2}
\textrm{~~for~}k\neq0,\textrm{~}\sF[F_{1}(\varphi_{1},I_{1})]_{0}=0.
\]
Obviously, the Fourier image of $K_{1}(\varphi,I)$ takes nonzero
values only for indices $(k,\ell)$ such that
$k\in\mathbb{Z}\setminus\{0\}$, $\ell\in\supp\sF[f]\setminus\{0\}$,
and
\[
\sF[K_{1}(\varphi,I)]_{(k,\ell)}
=i\ell\omega_{2}\,\sF[f(\varphi_{2})]_{\ell}\,
\sF[F_{1}(\varphi_{1},I_{1})]_{k}.
\]

Now we can describe the von Zeipel canonical transformation of the
first order. The resonant lattice is
$\mathbb{K}=\mathbb{Z}\,(\mu,-\nu)$, and, according to the general
scheme, one has
$\mathcal{K}_1(\psi,J)=\langle{}K_1(\psi,J)\rangle_{\mathbb{K}}$. With
the above computations, this immediately leads to
(\ref{eq:K1cal_psi_J}), (\ref{eq:def_beta}) and, consequently, to
(\ref{eq:K_trunc_1}).

An important piece of information is also the Poincar\'e-von~Zeipel
canonical transformation generated by $S_{(1)}(\varphi,J,\epsilon)$.
$S_{1}(\varphi,J)$ is a solution to the differential equation
$\omega\cdot\partial S_{1}/\partial\varphi=-K_{1}+\mathcal{K}_{1}$.
Seeking $S_{1}(\varphi,J)$ in the form
\begin{equation}
  S_{1}(\varphi,J)=2\Re\!\left(\,\sum_{k=1}^{\infty}
    \sF[f']_{k}G_{k}(\varphi_{1},J_{1})\, e^{ik\varphi_{2}}\right)
  \label{eq:S1eqG1Exp}
\end{equation}
one finally arrives at the countable system of equations
\[
\left(\frac{\partial}{\partial\varphi_{1}}
  +ik\lambda\right)\! G_{k}(\varphi_{1},J_{1})
= \lambda\sum_{n\in\mathbb{Z}\setminus\{0\},\textrm{~}
  n\neq-k\lambda}\frac{i^{n+1}}{2n}\,\beta(J_{1})^{|n|}\,
e^{in\varphi_{1}}\,,\ k \geq 1.
\]
For the solution we choose
\begin{equation}
  G_{k}(\varphi_{1},J_{1})
  =\lambda\sum_{n\in\mathbb{Z}\setminus\{0\},\textrm{~}
    n\neq-k\lambda}\frac{i^{n}}{2n(n+k\lambda)}\,\beta(J_{1})^{|n|}\,
  e^{in\varphi_{1}}.
  \label{eq:Gk}
\end{equation}
Of course, if $k\lambda\notin\mathbb{Z}$ then the restriction
$n\neq-k\lambda$ is void. On the other hand, if
$k\lambda\in\mathbb{Z}$, and this happens if and only if
$k\in\mathbb{Z}\nu$, then the solution $G_{k}(\varphi_{1},J_{1})$ is
not unique.

\subsection{Proofs of Theorem~\ref{thm:AsymptHamDynB1} and
  Proposition~\ref{thm:ConclusionResonant}}
\label{sec:1st_order_dynamics}

The purpose of this subsection is to give some details concerning the
dynamics generated by the Hamiltonian
$\mathcal{K}_{(1)}(\psi,J,\epsilon)$, as defined in
(\ref{eq:K_trunc_1}), and the proofs of
Theorem~\ref{thm:AsymptHamDynB1} and
Proposition~\ref{thm:ConclusionResonant}.

We start from the reduction to a two-dimensional Hamiltonian
subsystem. Since $\mu$ and $\nu$ are coprime there exist
$\tilde{\mu},\tilde{\nu}\in\mathbb{Z}$ such that
$\tilde{\mu}\mu+\tilde{\nu}\nu=1$. Put
\[
\mathbf{R}=\left(\begin{array}{cc}
\mu & -\nu\\
\tilde{\nu} & \tilde{\mu}\end{array}\right)
\]
and consider the canonical transformation $\chi=\mathbf{R}\psi$,
$J=\mathbf{R}^{\mathrm{T}}L$ generated by the function
$L\cdot\mathbf{R}\psi$. In particular,
\begin{displaymath}
  \chi_{1}=\mu\psi_{1}-\nu\psi_{2},~L_{2}=\nu J_{1}+\mu J_{2},~
  J_{1}=\mu L_{1}+\tilde{\nu}L_{2}.
\end{displaymath}
Since $\mathcal{K}_1(\psi,J)$ in (\ref{eq:K1cal_psi_J}) depends only
on the angle $\chi_1$, and not on $\chi_2$, the action $L_2$ is an
integral of motion. Let us define
\[
\mathcal{Z}(\chi_{1},J_{1})
=\varepsilon\mu\,\mathcal{K}_{1}(\mathbf{R}^{-1}\chi,J).
\]
Then
\begin{eqnarray*}
  \chi_{1}'(t) & = & \varepsilon\,
  \frac{\partial\mathcal{K}_{1}(\psi,J)}{\partial J_{1}}\,
  \frac{\partial J_{1}}{\partial L_{1}}\,
  =\,\frac{\partial\mathcal{Z}(\chi_{1},J_{1})}{\partial J_{1}}\,,\\
  J_{1}'(t) & = & -\varepsilon\,
  \frac{\partial\mathcal{K}_{1}(\psi,J)}{\partial\psi_{1}}\,
  = \,-\frac{1}{\mu}\,
  \frac{\partial\mathcal{Z}(\chi_{1},J_{1})}{\partial\chi_{1}}\,
  \frac{\partial\chi_{1}}{\partial\psi_{1}}\,
  = \,-\frac{\partial\mathcal{Z}(\chi_{1},J_{1})}{\partial\chi_{1}}\,.
\end{eqnarray*}
Thus the time evolution in coordinates $\chi_{1}$, $J_{1}$ is governed
by the Hamiltonian $\mathcal{Z}(\chi_{1},J_{1})$.

From the form of of the series (\ref{eq:K1cal_psi_J}) one can see that
the Hamiltonian $\mathcal{Z}(\chi_{1},J_{1})$ can be expressed, in a
compact way, in terms of a holomorphic function. Set
\begin{equation}
  h(z) = -\varepsilon\mu\omega_{1}\sum_{n=1}^{\infty}
  \sF[f]_{-n\nu}\, i^{n\mu}\, z^{n}
  \label{eq:defhz}
\end{equation}
and
\[
\varrho(x)
= \beta(x)^{\mu}
= \left(\frac{x}{x+p_{\theta}}\right)^{\!\mu/2}\!,\textrm{~~}x>0.
\]
Then, by assumption (\ref{eq:SumkFfk}), $h(z)$ is holomorphic on the
open unit disk $B_{1}\subset\mathbb{C}$ and
$h\in{}C^{1}(\overline{B_{1}})$. Moreover, one can write
$\mathcal{Z}(\chi_{1},J_{1})=\Re[h(\varrho(J_{1})e^{i\chi_{1}})]$
(note that $\sF[f]_{n\nu}=\overline{\sF[f]_{-n\nu}}$ since $f$ is
real). The Hamiltonian equations of motion read
\begin{equation}
  \chi_{1}'(t)=\frac{\varrho'(J_{1})}{\varrho(J_{1})}\,\Re[z\, h'(z)],
  \textrm{~}J_{1}'(t)
  =\Im[z\, h'(z)],\textrm{~with~}
  z=\varrho(J_{1})e^{i\chi_{1}}.
  \label{eq:HamEqChiJ}
\end{equation}
Concerning the asymptotic behavior of Hamiltonian trajectories
$(\chi_{1}(t),J_{1}(t))$, as $t\to+\infty$, one can formulate a
proposition under somewhat more general circumstances, as it is done
in Theorem~\ref{thm:AsymptHamDynB1}. \phantom{[}

\def\proof{\par{\it Proof of Theorem~\ref{thm:AsymptHamDynB1}}.
  \ignorespaces}
\begin{proof}
  Set $R(z)=\Re[h(z)]$, $z\in\overline{B_{1}}$. Then
  \begin{displaymath}
    \dd{}R_{z}\equiv(\Re[h'(z)],-\Im[h'(z)]).
  \end{displaymath}
  Hence $\dd{}R_{z}=0$ iff $h'(z)=0$, and the set of critical points
  of $R$ in $B_{1}$ has no accumulation points in $B_{1}$ and is at
  most countable. By Sard's theorem, almost all $y\in\mathbb{R}$ are
  regular values of $R|\partial B_{1}$.  If $y$ is a regular value
  both of $R$ and $R|\partial B_{1}$ then the level set $R^{-1}(y)$ is
  a compact one-dimensional $C^{1}$ submanifold with boundary in
  $\overline{B_{1}}$, $\partial R^{-1}(y)=R^{-1}(y)\cap\partial B_{1}$
  and $R^{-1}(y)$ is not tangent to $\partial B_{1}$ at any point.
  Moreover, $R^{-1}(y)\cap B_{1}$ is a smooth submanifold of $B_{1}$
  \cite{GuilleminPollack}. By the classification of compact connected
  one-dimensional manifolds \cite{GuilleminPollack}, every component
  of $R^{-1}(y)$ is diffeomorphic either to a circle or to a closed
  interval. But the first possibility is excluded because $R(z)$ is a
  harmonic function. In fact, if $\overline{U}\subset B_{1}$, $U$ is
  an open set, $\partial U\simeq S^{1}$ is a smooth submanifold of
  $B_{1}$ and $R(z)$ is constant on $\partial U$ then $R(z)$ is
  constant on $U$ and so is $h(z)$. Consequently, $h(z)$ is constant
  on $B_{1}$, a contradiction with our assumptions. Thus every
  component $\Gamma$ of $R^{-1}(y)$ is diffeomorphic to a closed
  interval, $\partial\Gamma=\{a,b\}=\Gamma\cap\partial B_{1}$, and
  $\Gamma$ is tangent to $\partial B_{1}$ neither at $a$ nor at $b$.
  Let $z\in B_{1}$ be such that $\dd R_{z}\neq0$. By the local
  submersion theorem \cite{GuilleminPollack}, $R$ is locally
  equivalent at $z$ to the canonical submersion
  $\mathbb{R}^{2}\to\mathbb{R}$. Hence $z$ possesses an open
  neighborhood $U$ such that $R(U)$ is an open interval. We know that
  almost every $y\in R(U)$ is a regular value both of $R$ and
  $R|\partial B_{1}$. By the Fubini theorem, for almost every
  $w\in{}U$, $R(w)$ is a regular value both of $R$ and
  $R|\partial{}B_{1}$. The same claim is true for almost all
  $w\in{}B_{1}$ because the set of critical points of $R$ in $B_{1}$
  is at most countable. It follows that for almost all
  $(\chi_{1},J_{1})\in\mathbb{R}\times\,]0,+\infty[\,$,
  $R(\varrho(J_{1})e^{i\chi_{1}})\neq R(0)$ is a regular value both of
  $R$ and $R|\partial B_{1}$.

  Suppose now that an initial condition $(\chi_{1}(0),J_{1}(0))$ has
  been chosen so that
  $y=R(\varrho(J_{1}(0))e^{i\chi_{1}(0)})\neq{}R(0)$ is a regular
  value both of $R$ and $R|\partial B_{1}$. Let $\Gamma$ be the
  component of $R^{-1}(y)$ containing the point
  $\varrho(J_{1}(0))e^{i\chi_{1}(0)}$.  Since the Hamiltonian
  $\mathcal{Z}(\chi_{1},J_{1})$ is an integral of motion the
  Hamiltonian trajectory $z(t)=\varrho(J_{1}(t))e^{i\chi_{1}(t)}$ is
  constrained to the submanifold $\Gamma\subset\overline{B_{1}}$. We
  have to show that $z(t)$ reaches the boundary $\partial B_{1}$ as
  $t\to+\infty$. The tangent vector to the trajectory at the point
  $z(t)$ equals
  \[
  \frac{\dd z(t)}{\dd t}
  = i\varrho(J_{1}(t))\varrho'(J_{1}(t))\,\overline{h'(z(t))}\,.
  \]
  Since $0\notin\Gamma$, $\varrho'(J_{1})>0$ for all $J_{1}>0$ and
  $h'(z)$ has no zeroes on $\Gamma$ (because $y$ is a regular value)
  it follows that $z(t)$ leaves any compact subset of $B_{1}$ in a
  finite time. It remains to show that $z(t)$ does not reach
  $\partial{}B_{1}$ in a finite time. But by equations of motion
  (\ref{eq:HamEqChiJ}),
  $|J_{1}'(t)|\leq\max_{z\in\partial{}B_{1}}|h'(z)|$ and so $J_{1}(t)$
  cannot grow faster than linearly.

  This reasoning shows (\ref{eq:limInfChiJ}). From
  (\ref{eq:HamEqChiJ}) and (\ref{eq:limInfChiJ}) it follows
  (\ref{eq:limInfDerChiJ}); one has only to justify the sign of the
  limit. Obviously, the limit must be nonnegative. Denote
  $\partial{}R=R|\partial B_{1}$. Then $\partial R$ can be regarded as
  a function of the angle variable, $\partial R(x)=\Re[h(e^{ix})]$,
  and one has
  \[
  (\partial R)'\big(\chi_{1}(\infty)\big)
  = -\Im\!\left[e^{i\chi_{1}(\infty)}\,
    h'\!\left(e^{i\chi_{1}(\infty)}\right)\right]\neq0
  \]
  because $y=\partial R(\chi_{1}(\infty))$ is a regular value of
  $\partial R$.
\end{proof}
\def\proof{\par{\it Proof}. \ignorespaces}

Finally, to derive Proposition~\ref{thm:ConclusionResonant} one has to
apply the inverted canonical transformation, from $(\psi,J)$ to
$(\varphi,I)$, generated by
$S_{(1)}(\varphi,J,\epsilon)=\varphi\cdot{}J+\epsilon{}S_{1}(\varphi,J)$.
Hence
\begin{equation}
  \label{eq:transf_Jpsi_Ivarphi}
  \psi = \varphi
  +\epsilon\,\frac{\partial S_{1}(\varphi,J)}{\partial J}\,,
  \textrm{~}I = J+\epsilon\,
  \frac{\partial S_{1}(\varphi,J)}{\partial\varphi}\,.
\end{equation}
For the proof we need a couple of auxiliary estimates.

\begin{lemma}
  \label{thm:SumBetaLeqLog}
  For all $\beta$, $0\leq\beta<1$, and all
  $a\in\mathbb{R}\setminus\mathbb{Z}$,
  \begin{equation}
    \sup_{j\in\mathbb{Z}}\,\sum_{n\in\mathbb{Z}}\,
    \frac{\beta^{|n+j|}}{|n-a|}
    \leq \frac{1}{\dist(a,\mathbb{Z})}+2+6|\log(1-\beta)|,
    \label{eq:SupjSumbetaaI}
  \end{equation}
  and for all $a\in\mathbb{Z}$,
  \begin{equation}
    \sup_{j\in\mathbb{Z}}\,\sum_{n\in\mathbb{Z},\textrm{~}
      n\neq a}\,\frac{\beta^{|n+j|}}{|n-a|}
    \leq 1+3|\log(1-\beta)|.
    \label{eq:SupjSumbetaaII}
  \end{equation}
\end{lemma}

\begin{proof}
  Notice that inequality (\ref{eq:SupjSumbetaaI}) is invariant if $a$
  is replaced either by $-a$ or by $k+a$, $k\in\mathbb{Z}$. Thus we
  can restrict ourselves to the interval $0<a\leq1/2$. Then
  $|a|=\dist(a,\mathbb{Z})$ and $|n-a|\geq|n|/2$.  This observation
  reduces (\ref{eq:SupjSumbetaaI}) to (\ref{eq:SupjSumbetaaII}) with
  $a=0$. Similarly, inequality (\ref{eq:SupjSumbetaaII}) is invariant
  if $a$ is replaced by $k+a$, $k\in\mathbb{Z}$. It follows that, in
  both cases, it suffices to show (\ref{eq:SupjSumbetaaII}) for $a=0$
  and with the range of $j$ restricted to nonnegative integers.

  Suppose that $j\geq0$. Splitting the range of summation in $n$ into
  the subranges $n\leq-j-1$, $-j\leq n\leq-1$ and $1\leq n$, one gets
  \[
  \sum_{n\in\mathbb{Z},\textrm{~}n\neq0}\,\frac{\beta^{|n+j|}}{|n|}
  \leq2|\log(1-\beta)|+\sum_{m=1}^{j}\frac{\beta^{j-m}}{m}\,.
  \]
  Furthermore,
  \[
  \sum_{1\leq m\leq j/2}\frac{\beta^{j-m}}{m}
  \leq\sum_{1\leq m\leq j/2}\frac{\beta^{m}}{m}\leq|\log(1-\beta)|
  \]
  and
  \[
  \sum_{j/2<m\leq j}\frac{\beta^{j-m}}{m}
  \leq\sum_{j/2<m\leq j}\frac{1}{m}\leq1.
  \]
  This shows the lemma.
\end{proof}

\def\proof{\par{\it Proof of
    Proposition~\ref{thm:ConclusionResonant}}.  \ignorespaces}
\begin{proof}
  If (\ref{eq:ResonantCase}) is true then $h(z)$ defined in
  (\ref{eq:defhz}) obeys the assumptions of
  Theorem~\ref{thm:AsymptHamDynB1} and so for almost all initial
  conditions $(\chi_{1}(0),J_{1}(0))$, equalities
  (\ref{eq:limInfChiJ}) and (\ref{eq:limInfDerChiJ}) hold. In
  particular, Theorem~\ref{thm:AsymptHamDynB1} implies that
  \begin{equation}
    1-\beta(J_{1}(t)) = O(t^{-1})\textrm{~~as~}t\to+\infty.
    \label{eq:AsymptBeta}
  \end{equation}
  Putting $f_{\nu}(\varphi)=\langle f(\varphi)\rangle_{\mathbb{Z}\nu}$
  one also has
  \begin{equation}
    \lim_{t\to+\infty}\frac{J_{1}(t)}{t}
    = \Im\!\left[e^{i\chi_{1}(\infty)}\,
      h'\!\left(e^{i\chi_{1}(\infty)}\right)\right]
    = -\frac{\varepsilon\omega_{2}}{2}\,
    f_{\nu}'\!\left(-\frac{\chi_{1}(\infty)}{\nu}
      -\frac{\pi\lambda}{2}\right) > 0.
    \label{eq:LimJ1Overt}
  \end{equation}

  Further, using (\ref{eq:Gk}) one can estimate (recalling that
  $\lambda=\mu/\nu$)
  \[
  \left|\frac{\partial G_{k}(\varphi_{1},J_{1})}{\partial J_{1}}\right|
  \leq \frac{\lambda}{2}\,\beta'(J_{1})
  \sum_{n\in\mathbb{Z}\setminus\{0\},\textrm{~}
    n \neq -k\lambda}\frac{1}{|n+k\lambda|}\,\beta(J_{1})^{|n|-1}
  \leq \frac{\mu\beta'(J_{1})}{1-\beta(J_{1})}\,.
  \]
  Hence, using (\ref{eq:def_beta}),
  \begin{equation}
    \left|\frac{\partial G_{k}(\varphi_{1},J_{1})}
      {\partial J_{1}}\right|
    \leq \frac{\mu}{2p_\theta}\,\frac{1-\beta(J_{1})}{\beta(J_{1})}\,.
    \label{eq:EstimDerGkJ1}
  \end{equation}

  As a next step let us estimate
  $|\partial{}G_{k}/\partial\varphi_{1}|$. Clearly,
  \[
  \left|\frac{\partial G_{k}(\varphi_{1},J_{1})}
    {\partial\varphi_{1}}\right|
  \leq\frac{\lambda}{2}\,\sum_{n\in\mathbb{Z}\setminus\{0\},\textrm{~}
    n\neq -k\lambda}\frac{1}{|n+k\lambda|}\,\beta(J_{1})^{|n|}\,.
  \]
  Writing $k\lambda=-j-a$, with $j\in\mathbb{Z}$ and $a=s/\nu$,
  $s=0,1,\ldots,\nu-1$, one can apply Lemma~\ref{thm:SumBetaLeqLog} to
  show that
  \begin{equation}
    \left|\frac{\partial G_{k}(\varphi_{1},J_{1})}
      {\partial\varphi_{1}}\right|
    \leq c'+c''|\log\!\left(1-\beta(J_{1})\right)|
    \label{eq:EstimDerGkVarphi1}
  \end{equation}
  where the constants $c'$, $c''$ do not depend on $k$ and
  $\varphi_{1}$, $J_{1}$.

  From definitions (\ref{eq:S1eqG1Exp}), (\ref{eq:Gk}) and from
  assumption (\ref{eq:SumkFfk}) one can readily see that
  $S(\varphi,J)$ is $C^{1}$ in $\varphi_{2}$ and $C^{\infty}$ in
  $\varphi_{1}$, $J_{1}$ (and does not depend on $J_{2}$). Moreover,
  from (\ref{eq:EstimDerGkJ1}) and (\ref{eq:AsymptBeta}) it follows
  that
  \begin{equation}
    \frac{\partial S_{1}(\varphi(t),J(t))}{\partial J_{1}}
    = O(t^{-1})\textrm{~~as~}t\to+\infty.
    \label{eq:der_S1_J1}
  \end{equation}
  Similarly, estimate (\ref{eq:EstimDerGkVarphi1}) implies
  \begin{equation}
    \frac{\partial S_{1}(\varphi(t),J(t))}{\partial\varphi_{1}}
    = O(\log(t))\textrm{~~as~}t\to+\infty.
    \label{eq:der_S1_varphi1}
  \end{equation}

  Now one can deduce the asymptotic behavior of $\varphi_{1}(t)$ and
  $I_{1}(t)$. Observe that $\psi_{2}'=\omega_{2}$ and
  $S_{1}(\varphi,J)$ does not depend on $J_{2}$, hence
  $\psi_{2}(t)=\varphi_{2}(t)=\omega_{2}t$. Furthermore,
  $\psi_{1}=(\chi_{1}+\nu\psi_{2})/\mu$ and so
  \begin{equation}
    \label{eq:lim_psi1}
    \lim_{t\to+\infty}\left(\psi_{1}(t)-\omega_{1}t\right)
    = \frac{1}{\mu}\,\chi_{1}(\infty).
  \end{equation}
  Put $\phi(\infty)=\chi_{1}(\infty)/\mu$. To conclude the proof it
  suffices to recall the transformation rules
  (\ref{eq:transf_Jpsi_Ivarphi}) and to take into account
  (\ref{eq:LimJ1Overt}) jointly with (\ref{eq:der_S1_varphi1}) and
  (\ref{eq:lim_psi1}) jointly with (\ref{eq:der_S1_J1}).
\end{proof}
\def\proof{\par{\it Proof}. \ignorespaces}

\section{Proofs of Proposition~\ref{thm:F} and
  Proposition~\ref{thm:asym_phit_F_lin}}
\label{sec:SimplifiedEquation}



Analyzing equation (\ref{eq:ODEsimplifiedF}) we prefer to work with a
rescaled time (or one can choose the units so that $b=1$) and,
simplifying the notation, we consider a differential equation of the
form
\begin{equation}
  g'(t) = \frac{\varrho(t)}{g(t)+\sqrt{g(t)^{2}-a(t)^{2}}\sin(t+\phi)}
  \label{eq:ODEgtvarphi}
\end{equation}
where $\varrho(t)$, $a(t)$ are continuously differentiable real
functions, $a(t)$ is strictly positive and $\phi$ is a real constant.
In the resonant case the functions $\varrho(t)$, $a(t)$ are supposed
to be $2\pi$-periodic which means for the original data that
$\Omega\in\mathbb{N}$ (and $b=1$).

In the first step we estimate the growth of a solution on an interval
of length $\pi/2$. Let $\|f\|=\max|f(t)|$ denote the norm in
$C([0,\pi/2])$, and put
\[
A=\underset{0\leq t\leq\pi/2}{\min}\, a(t)>0.
\]
Consider for a moment the differential equation
\begin{equation}
  h'(t)=\frac{\varrho(t)}{h(t)-\sqrt{h(t)^{2}-a(t)^{2}}\,\cos(t)}
  \label{eq:ODEhtOn0Pihalf}
\end{equation}
on the interval $[0,\pi/2]$ with an initial condition
$h(0)=h_{0}>\|a\|$.  The goal is to show that for large values of
$h_{0}$, an essential contribution to the growth of a solution $h(t)$
on this interval comes from a narrow neighborhood of the point $t=0$.

\begin{remark}
  \label{thm:RemExistUniquephi}
  If $\varrho(t)$ is nonnegative on the interval $[0,\pi/2]$ then a
  solution $h(t)$ to (\ref{eq:ODEhtOn0Pihalf}) surely exists and is
  unique. In the general case, the existence and uniqueness is
  guaranteed provided the initial condition $h_{0}$ is sufficiently
  large. From (\ref{eq:ODEhtOn0Pihalf}) one derives that
  $|h'(t)|\leq2\|\varrho\|h(t)/A^{2}$ and so
  \begin{equation}
    \exp\!\left(-2\|\varrho\|t/A^{2}\right)h_{0}
    \leq h(t)\leq\exp\!\left(2\|\varrho\|t/A^{2}\right)h_{0}
    \label{eq:EstimhtByh0}
  \end{equation}
  as long as $h(t)$ makes sense. Consequently, a sufficient condition
  for existence of a solution is
  $h_{0}>\exp(\pi\|\varrho\|/A^{2})\|a\|$.
\end{remark}

\begin{lemma}
  \label{thm:Preparatoryht}
  Let $\varrho,a\in C^{1}([0,\pi/2])$ be real functions,
  $\varrho(0)\neq0$ and $a(t)>0$ on $[0,\pi/2]$. Consider the set of
  solutions $h(t)$ to the differential equation
  (\ref{eq:ODEhtOn0Pihalf}) on the interval $[0,\pi/2]$ with a
  variable initial condition $h(0)=h_{0}$ for $h_{0}$ sufficiently
  large. Then
  \[
  h\!\left(\frac{\pi}{2}\right)
  = h_{0}+\frac{\pi\varrho(0)}{a(0)}
  +O\!\left(h_{0}^{\,-1}\log(h_{0})\right)\textrm{~~as~}h_{0}\to+\infty.
  \]
\end{lemma}

\begin{proof}
  Let us fix $\eta$, $0<\eta\leq\pi/2$, so that $|\varrho(t)|>0$ on
  the interval $[0,\eta\,[\,$, i.e. $\varrho(t)$ does not change its
  sign on that interval. Thus any solution $h(t)$ to
  (\ref{eq:ODEhtOn0Pihalf}) is strictly monotone on $[0,\eta]$. For
  $\eta\leq t\leq\pi/2$ one can estimate $|h'(t)|\leq C/h(t)$ where
  $C=\|\varrho\|/(1-\cos(\eta))$.  In view of (\ref{eq:EstimhtByh0})
  it follows that
  \begin{equation}
    h(\pi/2)-h(\eta)=O\left(h_{0}^{\;-1}\right)\textrm{~as~}
    h_{0}\to+\infty.\label{eq:hPihalfMinusheta}
  \end{equation}

  Set
  \[
  h_{1}=\min\{h(0),h(\eta)\},\mbox{~}h_{2}
  =\max\{h(0),h(\eta)\},\mbox{~}\Delta=h(\eta)-h_{0}.
  \]
  Then $|\Delta|=h_{2}-h_{1}$. Set, for $x\geq a>0$,
  \[
  \Psi(x,a,t)=\frac{1}{x-\sqrt{x^{2}-a^{2}}\,\cos(t)}\,.
  \]
  One has $h'(t)=\varrho(t)\,\Psi(h(t),a(t),t)$. If
  $x\geq2u/\sqrt{3}\,>0$ then $\sqrt{x^{2}-u^{2}}\geq x/2$ and
  \begin{eqnarray}
    \left|\frac{\partial}{\partial x}\Psi(x,u,t)\right|
    & = & \frac{\left|\sqrt{x^{2}-u^{2}}-x\cos(t)\right|}
    {\sqrt{x^{2}-u^{2}}\left(x-\sqrt{x^{2}-u^{2}}\cos(t)\right)^{\!2}}\,
    \leq\,\frac{2}{x}\,\Psi(x,u,t),\label{eq:derx_Psi}\\
    \noalign{\medskip}
    \left|\frac{\partial}{\partial u}\Psi(x,u,t)\right|
    & = & \frac{u\cos(t)}{\sqrt{x^{2}-u^{2}}
      \left(x-\sqrt{x^{2}-u^{2}}\cos(t)\right)^{\!2}}\,
    \leq\,\frac{3}{u}\,\Psi(x,u,t).
    \label{eq:deru_Psi}
  \end{eqnarray}
  Observe also that, for $x\geq u>0$,
  \begin{equation}
    \int_{0}^{\pi/2}\Psi(x,u,t)\,\mbox{d}t
    =\frac{2}{u}\,\arctan\!\left(\frac{x+\sqrt{x^{2}-u^{2}}}{u}\right)
    \leq\frac{\pi}{u}\,.
    \label{eq:int_Psi_estim}
  \end{equation}

  Assuming that $h_{0}$ is sufficiently large and using
  (\ref{eq:derx_Psi}), (\ref{eq:int_Psi_estim}) one can estimate
  \begin{eqnarray*}
    \left|\Delta-\int_{0}^{\eta}\varrho(t)\Psi(h_{0},a(t),t)\,
      \mbox{d}t\right| & \leq & \int_{0}^{\eta}|\varrho(t)|
    \left|\Psi(h(t),a(t),t)-\Psi(h_{0},a(t),t)\right|\mbox{d}t\\
    & \leq & \frac{2\|\varrho\|}{h_{1}}\int_{0}^{\pi/2}
    \left(\int_{h_{1}}^{h_{2}}\Psi(x,A,t)\,\mbox{d}x\right)\mbox{d}t\\
    & \leq & \frac{2\pi\|\varrho\|}{Ah_{1}}\,|\Delta|.
  \end{eqnarray*}
  In view of (\ref{eq:EstimhtByh0}) it follows that
  \[
  \Delta=\left(1+O(h_{0}^{\,-1})\right)\int_{0}^{\eta}
  \varrho(t)\Psi(h_{0},a(t),t)\,\mbox{d}t.
  \]
  Furthermore, with the aid of (\ref{eq:deru_Psi}) one finds that
  \[
  \left|\varrho(t)\Psi(h_{0},a(t),t)-\varrho(0)\Psi(h_{0},a(0),t)\right|
  \leq C'\Psi(h_{0},A,t)t
  \]
  where
  \[
  C'=\left(1+\frac{9\|\varrho\|^{2}}{A^{2}}\right)^{\!1/2}
  \sqrt{\|\varrho'\|^{2}+\|a'\|^{2}}\,.
  \]
  Note that
  \[
  \int_{0}^{\eta}\Psi(h_{0},A,t)\, t\,\mbox{d}t
  \leq\frac{\pi}{2}\int_{0}^{\pi/2}\frac{\sin(t)}{h_{0}
    -\sqrt{h_{0}^{\,2}-A^{2}}\cos(t)}\,\mbox{d}t
  = O\!\left(h_{0}^{\,-1}\log(h_{0})\right)
  \]
  and
  \begin{eqnarray*}
    \int_{0}^{\eta}\varrho(0)\Psi(h_{0},a(0),t)\,\mbox{d}t
    & = & \varrho(0)\int_{0}^{\pi/2}\Psi(h_{0},a(0),t)\,
    \mbox{d}t+O\!\left(h_{0}^{\,-1}\right)\\
    & = & \frac{2\varrho(0)}{a(0)}\,
    \arctan\!\left(\frac{h_{0}+\sqrt{h_{0}^{\,2}-a(0)^{2}}}{a(0)}\right)
    +O\!\left(h_{0}^{\,-1}\right)\\
    & = & \frac{\pi\varrho(0)}{a(0)}+O\!\left(h_{0}^{\,-1}\right).
  \end{eqnarray*}
  Altogether this means that
  \[
  h(\eta)-h_{0}=\Delta=\frac{\pi\varrho(0)}{a(0)}
  +O\!\left(h_{0}^{\,-1}\log(h_{0})\right).
  \]
  Recalling (\ref{eq:hPihalfMinusheta}), the lemma follows.
\end{proof}

Consider the mapping $\mathcal{H}:h(0)\mapsto h(\pi/2)$, where $h(t)$
runs over solutions to the differential equation
(\ref{eq:ODEhtOn0Pihalf}).  From the general theory of ordinary
differential equations it is known that $\mathcal{H}$ is a $C^{1}$
mapping well defined on a neighborhood of $+\infty$.
Lemma~\ref{thm:Preparatoryht} claims that
$\mathcal{H}(x)=x+\pi\varrho(0)/a(0)+O(x^{-1}\log(x))$. On the basis
of similar arguments, the inverse mapping
$\mathcal{H}^{-1}:h(\pi/2)\mapsto{}h(0)$ is also well defined and
$C^{1}$ on a neighborhood of $+\infty$. From the asymptotic behavior
of $\mathcal{H}(x)$ one readily derives that
$\mathcal{H}^{-1}(y)=y-\pi\varrho(0)/a(0)+O(y^{-1}\log(y))$. These
considerations make it possible to reverse the roles of the boundary
points $0$ and $\pi/2$. Moreover, splitting the interval $[0,2\pi]$
into four subintervals of length $\pi/2$ one arrives at the following
lemma.

\begin{lemma}
  \label{thm:Preparatoryht02Pi}
  Let $\varrho,a\in C^{1}([0,2\pi])$ be real functions,
  $\varrho(\pi)\neq0$ and $a(t)>0$ on $[0,2\pi]$.  Consider the set of
  solutions $h(t)$ to the differential equation
  \[
  h'(t) = \frac{\varrho(t)}{h(t)+\sqrt{h(t)^{2}-a(t)^{2}}\,\cos(t)}
  \]
  on the interval $[0,2\pi]$ with a variable initial condition
  $h(0)=h_{0}$ for $h_{0}$ sufficiently large. Then
  \[
  h(2\pi) = h_{0}+\frac{2\pi\varrho(\pi)}{a(\pi)}
  +O\!\left(h_{0}^{\,-1}\log(h_{0})\right)\text{~~as~}h_{0}\to+\infty.
  \]
\end{lemma}

In the next step, applying repeatedly
Lemma~\ref{thm:Preparatoryht02Pi} one can show that solutions of the
differential equation (\ref{eq:ODEgtvarphi}) in the resonant case
$\Omega\in\mathbb{N}$ (with $b=1$) grow with time linearly provided
the initial condition is sufficiently large and the phase $\phi$
belongs to a certain interval.

\begin{proposition}
  \label{thm:ResonantWb}
  Suppose $\varrho(t)$, $a(t)$ are continuously differentiable
  $2\pi$-periodic real functions, $\phi\in\mathbb{R}$, $a(t)$ is
  everywhere positive and
  \[
  \varrho\!\left(-\phi-\frac{\pi}{2}\right)>0.
  \]
  Let $g(t)$ be a solution of the differential equation
  (\ref{eq:ODEgtvarphi}) on the interval $t\geq0$ with the initial
  condition $g(0)=g_{0}\geq1$.  If $g_{0}$ is sufficiently large then
  \[
  g(t)=\frac{\varrho\!\left(-\phi-\frac{\pi}{2}\right)}
  {a\!\left(-\phi-\frac{\pi}{2}\right)}\, t
  +O\!\left(\log(t)^{2}\right)\textrm{~~as~}t\to+\infty.
  \]
\end{proposition}

\def\proof{\par{\it Proof of Proposition~\ref{thm:F}}.  \ignorespaces}
\begin{proof}
  This is just an immediate corollary of
  Proposition~\ref{thm:ResonantWb}. It suffices to go back to the
  original notation and equation (\ref{eq:ODEsimplifiedF}) by applying
  the substitution $F(t)=g(bt)$, $a(t)=\tilde{a}(bt)$,
  $\varrho(t)=\tilde{a}(t)\tilde{a}'(t)$. Equation
  (\ref{eq:ODEsimplifiedF}) transforms into (\ref{eq:ODEgtvarphi})
  (with $a(t)$ being replaced by $\tilde{a}(t)$), and
  Proposition~\ref{thm:ResonantWb} is directly applicable and gives
  the result.
\end{proof}
\def\proof{\par{\it Proof}. \ignorespaces}


\def\proof{\par{\it Proof of Proposition~\ref{thm:asym_phit_F_lin}}.
  \ignorespaces}
\begin{proof}
  Put $\zeta(t)=bt-\phi(t)$. Recall that $\phi'(t)$ obeys the estimate
  (\ref{eq:phi_der_estim}). Choose $t_{\ast}\in\mathbb{R}$ such that
  $|\phi'(t)|\leq{}b/2$, $\forall{}t\geq{}t_{\ast}$.
  Hence the function $\zeta(t)$ is strictly increasing and
  $b/2\leq\zeta'(t)\leq3b/2$.  Moreover, we choose $t_{\ast}$
  sufficiently large so that
  \begin{equation}
    F(t+s)\leq\sqrt{2}\, F(t)\ \mbox{for\ }0
    \leq s\leq3\pi b,\mbox{\ and\ }F(t)\geq\sqrt{2}\, A_{2},\ \ 
    \forall t\geq t_{\ast}.
    \label{eq:estim_tstar}
  \end{equation}
  Fix $\ell\in\mathbb{N}$ such that $\zeta(2\pi(\ell+1))\geq
  t_{\ast}$.  Put $\tau_{k}=\zeta(2\pi(\ell+k))$, $k\in\mathbb{N}$.
  Note that $\pi b\leq\tau_{k+1}-\tau_{k}\leq3\pi b$. For a given
  $k\in\mathbb{N}$ put
  \[
  F_{1}=\min_{t\in[\tau_{k},\tau_{k+1}]}F(t),\mbox{\ }F_{2}
  = \max_{t\in[\tau_{k},\tau_{k+1}]}F(t).
  \]
  One has
  \begin{eqnarray*}
    \int_{\tau_{k}}^{\tau_{k+1}}|\phi'(t)|\,\mbox{d}t
    & \leq & \frac{2}{b}
    \int_{\zeta(2\pi(\ell+k))}^{\zeta(2\pi(\ell+k+1))}|\phi'(t)|\zeta'(t)\,
    \mbox{d}t\\
    \noalign{\medskip}
    & \leq & \frac{2A_{2}A_{3}}{b\sqrt{F_{1}^{\,2}-A_{2}^{\,2}}}
    \int_{2\pi(\ell+k)}^{2\pi(\ell+k+1)}\frac{|\cos s|}
    {\tilde{F}(s)+\sqrt{\tilde{F}(s)^{2}-\tilde{a}(s)^{2}}\,
      \sin s}\,\mbox{d}s
  \end{eqnarray*}
  where $\tilde{F}=F\circ\zeta^{-1}$, $\tilde{a}=a\circ\zeta^{-1}$.
  Put $M_{+}=2\pi(\ell+k)+[\,0,\pi\,]$,
  $M_{-}=\linebreak2\pi(\ell+k)+[\,\pi,2\pi\,]$. One has
  \[
  \int_{M_{+}}\frac{|\cos s|}{\tilde{F}(s)+\sqrt{\tilde{F}(s)^{2}
      -\tilde{a}(s)^{2}}\,\sin s}\,\mbox{d}s
  \leq 2\int_{0}^{\pi/2}\frac{\cos s}{F_{1}
    +\sqrt{F_{1}^{\,2}-A_{2}^{\,2}}\,\sin s}\,\mbox{d}s
  \leq \frac{2\log2}{\sqrt{F_{1}^{\,2}-A_{2}^{\,2}}}.
  \]
  For $s\in M_{-}$ one can estimate
  \[
  \frac{1}{\tilde{F}(s)+\sqrt{\tilde{F}(s)^{2}
      -\tilde{a}(s)^{2}}\,\sin s}
  \leq \frac{F_{2}+\sqrt{F_{2}^{\,2}-A_{1}^{\,2}}\,|\sin s|}
  {F_{1}^{\,2}\cos^{2}s+A_{1}^{\,2}\sin^{2}s}
  < \frac{2}{F_{1}+\sqrt{F_{1}^{\,2}-A_{1}^{\,2}}\,\sin s}
  \]
  where we have used that
  \[
  \frac{F_{2}}{\sqrt{F_{2}^{\,2}-A_{1}^{\,2}}}
  \leq \frac{F_{1}}{\sqrt{F_{1}^{\,2}-A_{1}^{\,2}}}\,,\ F_{2}^{\,2}
  -A_{1}^{\,2}\leq2F_{1}^{\,2}-A_{1}^{\,2}
  \leq 4F_{1}^{\,2}-5A_{1}^{\,2}<4(F_{1}^{\,2}-A_{1}^{\,2}),
  \]
  as it follows from (\ref{eq:estim_tstar}). Thus one arrives at the
  estimates
  \begin{eqnarray*}
    \int_{M_{-}}\frac{|\cos s|}{\tilde{F}(s)
      +\sqrt{\tilde{F}(s)^{2}-\tilde{a}(s)^{2}}\,\sin s}\,\mbox{d}s
    & \leq & 4\int_{0}^{\pi/2}\frac{\cos s}{F_{1}
      -\sqrt{F_{1}^{\,2}-A_{1}^{\,2}}\,\sin s}\,\mbox{d}s\\
    & \leq & \frac{4}{\sqrt{F_{1}^{\,2}-A_{1}^{\,2}}}
    \log\!\left(\frac{2F_{1}^{\,2}}{A_{1}^{\,2}}\right)
  \end{eqnarray*}
  and
  \[
  \int_{\tau_{k}}^{\tau_{k+1}}|\phi'(t)|\,\mbox{d}t
  \leq \frac{32A_{2}A_{3}}{bF_{1}^{\,2}}\,
  \log\!\left(\frac{2F_{1}}{A_{1}}\right).
  \]
  Referring to the asymptotic behavior (\ref{eq:Ft_asympt}) one
  concludes that there exists a constant $C_{\ast}>0$ such that
  \[
  \int_{\tau_{1}}^{\infty}|\phi'(t)|\,\mbox{d}t
  = \sum_{k=1}^{\infty}\int_{\tau_{k}}^{\tau_{k+1}}|\phi'(t)|\,\mbox{d}t
  \leq C_{\ast}\sum_{j=\ell+1}^{\infty}\frac{\log(j)}{j^{2}}<\infty.
  \]
  Hence the limit
  $\lim_{t\to+\infty}\phi(t)=\phi(\infty)\in\mathbb{R}$ does exist and
  (\ref{eq:phit_asympt}) follows.
\end{proof}
\def\proof{\par{\it Proof}. \ignorespaces}


\section{Some derivations related to the guiding center
  coordinates\label{sec:Conclusion}}

We keep the notation introduced in
Subsection~\ref{sec:results_guiding} but for the sake of simplicity we
again put $m=e=1$. The physical constants can readily be reestablished
if necessary.

In particular, $X=q+(1/b)v^\perp$, $R=-(1/b)v^\perp$, and a direct
computation yields
\[
|R|^{2}-|X|^{2}=\frac{2a}{b}\,,\ |v|^{2}
=p_{r}^{\,2}+\frac{b^{2}r^{2}}{4}+\frac{a^{2}}{r^{2}}+ba.
\]
Using (\ref{eq:canonic_r}), (\ref{eq:canonic_pr}) one derives the
equalities
\begin{equation}
  |X|^{2} = \frac{1}{b}\,(2I+|a|-a),\ 
  |R|^{2} = \frac{1}{b}\,(2I+|a|+a)
  \label{eq:X_R_norm}
\end{equation}
and
\begin{equation}
  r=\left(|X|^{2}+|R|^{2}+2|X||R|\sin(\varphi)\right)^{\!1/2},\
  p_{r}=\frac{b}{r}\,|X||R|\cos(\varphi).
\label{eq:r_pr_rel_X_R}
\end{equation}
On the other hand, one has
\begin{equation}
  r^{2}=|X|^{2}+|R|^{2}+2X\cdot R
  =\frac{2}{b}\left(2I+|a|+2\sqrt{I(I+|a|)}\,
    \cos(\vartheta-\chi)\right)\!.
  \label{eq:r2_X_R_norm}
\end{equation}
By comparison of (\ref{eq:r2_X_R_norm}) with (\ref{eq:r_pr_rel_X_R})
one shows equality (\ref{eq:varth_rel_varphi_chi}).

Further we sketch derivation of the the asymptotic relations
(\ref{eq:chi_vtheta_asympt}). To this end, let us compute the
derivative $\chi'(t)$. This can be done by differentiating the
equality
\[
r\,(\cos\theta,\sin\theta)
=|X|(\cos\chi,\sin\chi)+|R|(\cos\vartheta,\sin\vartheta)
\]
and then taking the scalar product with the vector
$(-\sin\theta,\cos\theta)$. One has
\[
\theta'=\frac{\partial H}{\partial p_{\theta}}=\frac{a}{r^{2}}
+\frac{b}{2}
\]
where $H$ is the Hamiltonian (\ref{eq:Ham_polar_coord}) expressed in
polar coordinates. Using (\ref{eq:varth_rel_varphi_chi}) and some
straightforward manipulations one finally arrives at the differential
equation
\begin{equation}
  \chi'=\frac{|R|a'\cos\varphi}{|X|br^{2}}\,.\label{eq:ODE_chi}
\end{equation}

Equation (\ref{eq:ODE_chi}) admits an asymptotic analysis with the aid
of similar methods as those used in
Section~\ref{sec:SimplifiedEquation}. In order to spare some space we
omit the details. We still assume that $\Omega/b\in\mathbb{N}$.
Recalling (\ref{eq:I_vphi_asympt}), (\ref{eq:C_xi_asympt_def}) one
observes that the main contribution to the growth of $\chi(t)$ over a
period $T=2\pi/b$ equals
\begin{eqnarray*}
  \chi((n+1)T)-\chi(nT) & \sim & \frac{1}{4CnT}\,
  \lim_{\alpha\to0}\,\int_{0}^{T}
  \frac{a'(t)\cos(bt+\phi(\infty))}{1+\sqrt{1-\alpha^{2}}\,
    \sin(bt+\phi(\infty))}\,\mbox{d}t.
\end{eqnarray*}
Proceeding this way one finally concludes that the first equality in
(\ref{eq:chi_vtheta_asympt}) holds, with $\chi(\infty)$ being a real
constant and
\[
D=\frac{1}{4\pi f'(-\xi)}\,
\int_{0}^{\pi}\left(f'\!\left(\frac{\Omega}{b}\,
    t-\xi\right)-f'\!\left(-\frac{\Omega}{b}\, t-\xi\right)\right)
\frac{\sin(t)}{1-\cos(t)}\,\mbox{d}t.
\]
Being given the Fourier series
$f'(t)=\sum_{k=1}^{\infty}\left(a_{k}\cos(kt)+b_{k}\sin(kt)\right)$
one can also express
\begin{displaymath}
  D = \frac{1}{2}\,\sum_{k=1}^{\infty}\left(a_{k}\sin(k\xi)
    +b_{k}\cos(k\xi)\right)\bigg/
  \sum_{k=1}^{\infty}\left(a_{k}\cos(k\xi)-b_{k}\sin(k\xi)\right)
\end{displaymath}
as it follows from the equality, valid for any $n\in\mathbb{N}$,
\begin{displaymath}
  \int_0^\pi \frac{\sin(nt)\sin(t)}{1-\cos(t)}\,\mathrm{d}t = \pi.
\end{displaymath}
Moreover, (\ref{eq:I_vphi_asympt}), (\ref{eq:varth_rel_varphi_chi})
and the first equality in (\ref{eq:chi_vtheta_asympt}) imply the
second equality in (\ref{eq:chi_vtheta_asympt}) where one has to put
$\vartheta(\infty)=\phi(\infty)+\chi(\infty)-\pi/2$.


\end{document}